\documentclass[a4paper,12pt]{article}
\pdfoutput=1
\usepackage{test}

\usepackage{xcolor}
\definecolor{pantone}{RGB}{1, 33, 105}

\usepackage{hyperref}
\hypersetup{colorlinks,linkcolor=blue,citecolor=blue,urlcolor=pantone}
\usepackage{doi}
\usepackage{booktabs,caption,graphicx,array}
\usepackage{tikz}

\usepackage[style=ext-authoryear-comp,
sorting=nyt, 
sortcites=false, 
dashed=false, 
maxcitenames=2, 
maxbibnames=99, 
uniquelist=false,
uniquename=false,
giveninits=true, 
useprefix=true, 
natbib, 
date=year 
]{biblatex}

\AtBeginRefsection{\GenRefcontextData{sorting=ynt}}
\AtEveryCite{\localrefcontext[sorting=ynt]}

\DeclareFieldFormat{pages}{#1} 
\renewbibmacro{in:}{\ifentrytype{article}{}{\printtext{\bibstring{in}\intitlepunct}}} 

\DeclareFieldFormat[article,inbook,incollection,inproceedings,patent,thesis,unpublished]{titlecase:title}{\MakeSentenceCase*{#1}} 

\newbibmacro*{atoda:titlelink}[1]{%
  \ifhyperref
    {\iffieldundef{doi}
       {\iffieldundef{eprint}
          {\iffieldundef{url}
             {#1}
             {\href{\thefield{url}}{#1}}}
          {\href{https://arxiv.org/abs/\thefield{eprint}}{#1}}}
       {\href{https://doi.org/\thefield{doi}}{#1}}}
    {#1}}

\DeclareFieldFormat{title}{\usebibmacro{atoda:titlelink}{\mkbibemph{#1}}}
\DeclareFieldFormat
  [article,inbook,incollection,inproceedings,patent,thesis,unpublished]
  {title}{\usebibmacro{atoda:titlelink}{\mkbibquote{#1\isdot}}}

\renewbibmacro*{doi+eprint+url}{}

\usepackage[inline,shortlabels]{enumitem}
\setlist[enumerate,1]{label=(\roman*)}

\usepackage[margin = 1.25in]{geometry}
\usepackage{setspace}
\title{Land and Infinite Debt Rollover\thanks{We thank Joseph Stiglitz for continuous discussions on the topic.}}
\author{Tomohiro Hirano\thanks{Department of Economics, Royal Holloway, University of London; Center for Macroeconomics at the London School of Economics; Canon Institute for Global Studies. 
Email: \href{mailto:tomohiro.hirano@rhul.ac.uk}{tomohiro.hirano@rhul.ac.uk}.} \and Alexis Akira Toda \thanks{Department of Economics, Emory University and Research Institute for Economics and Business Administration, Kobe University. Email: \href{mailto:alexis.akira.toda@emory.edu}{alexis.akira.toda@emory.edu}.}}

\numberwithin{equation}{section}
\numberwithin{prop}{section}

\newcommand{\cK}{\mathcal{K}}
\newcommand{\cR}{\mathcal{R}}
\newcommand{\cT}{\mathcal{T}}
\newcommand{\modelF}{\ensuremath{\mathrm{F}}\xspace}
\newcommand{\modelBU}{\ensuremath{\mathrm{B}^{U}}\xspace}
\newcommand{\modelBR}{\ensuremath{\mathrm{B}^{R}}\xspace}

\begin{document}
\maketitle

\begin{abstract}
Many analyses of infinite debt rollover omit productive land, whereas prior work argues that land precludes rollover. We show that when land's income share remains bounded away from zero, every equilibrium is bubbleless and Pareto efficient, ruling out rollover. If instead the land share vanishes and the autarky interest rate is below rent growth, every equilibrium contains a land bubble and permits rollover despite Pareto efficiency. Iowa farmland data show falling yields and converging district prices as price--rent ratios rise. This pattern rejects a bubbleless Gordon model but is consistent with a rational bubble earning the return on land.

\medskip

		
\medskip

\end{abstract}

\section{Introduction}

Can a fiscal authority service its debt forever by issuing new debt, without eventually raising taxes or reducing expenditures? A growing literature studies infinite debt rollover when interest rates are low \citep{Blanchard2019,BallMankiw2023,Kocherlakota2023ECMA,BloiseVailakis2024,AbelPanageas2025,MianStraubSufi2025}. A standard benchmark traces back to \citet{Diamond1965}: an overlapping generations (OLG) economy may accumulate so much capital that the interest rate falls below the economic growth rate, creating both dynamic inefficiency and room for rolling over public debt. Yet these analyses typically study economies without productive land.

Omitting land may matter. Land is a productive, non-reproducible, and non-depreciating asset, and in a frictionless economy its return must equal the return on government debt. An older literature argues that land eliminates capital overaccumulation and restores Pareto efficiency \citep{McCallum1987,Homburg1991,Homburg2014}. The argument is easiest to see on a balanced growth path. If the land price grows at the same rate as output, its capital gain already equals the economic growth rate; the positive rental yield then makes the total return strictly higher, which makes infinite debt rollover unsustainable. \citet[p.~339]{Nichols1970} summarizes the argument: ``the rate of return on land must exceed the growth rate. [\ldots] Oversaving or inefficient accumulation of excess capital in the sense of Phelps is shown to be impossible.'' If this reasoning applied whenever productive land exists, the modern debate over infinite debt rollover would largely disappear. However, the land price and rent need not grow with output when production is not Cobb--Douglas. More generally, land may become economically inessential even though it remains technologically productive. With constant elasticity of substitution (CES) production, for example, an elasticity above one can make land's income share vanish as the economy grows. The mere presence of productive land therefore does not settle whether rollover is possible.

The modern literature on debt rollover and the older literature on land have thus developed largely in isolation, leaving the central theoretical question unresolved: does productive land preclude infinite debt rollover, or can the two coexist? The link between them is a possible bubble in the price of land---a component not backed by the present value of rents. In our economy, infinite debt rollover is possible if and only if land prices contain such a bubble, because government debt can replace part of the bubble without changing the allocation. This paper provides three contributions. First, if land's income share remains bounded away from zero, land prices are bubbleless, every equilibrium is Pareto efficient, and infinite debt rollover is impossible. Second, if the land share vanishes and the autarky interest rate lies below rent growth, every equilibrium contains a land bubble that government debt can replace. More surprisingly, under a regularity condition, every equilibrium is Pareto efficient. Thus, contrary to the common view, neither an asset price bubble nor infinite debt rollover requires dynamic inefficiency. Finally, we bring the theory to data using Iowa farmland. The data show the declining cash yield predicted by a bubbly equilibrium, and the cross-district pattern of land prices since the mid-1990s is consistent with a bubble that earns the return on land, as in our deterministic theory, but not with a Gordon model without a bubble, even when the capitalization rate is free to move.

We begin our analysis with a deterministic sequential economy in which labor and a fixed supply of land produce the consumption good. Agents may have arbitrary finite or infinite lifetimes. Because factor supplies are inelastic and production is static, the economy has an equivalent endowment representation in which agents receive non-asset income and land pays a dividend. This equivalence lets us study the roles of production, asset prices, and welfare separately.

Lemma~\ref{lem:debt_bubble} establishes a debt--bubble equivalence that adapts known bubble-substitution arguments. Because government bonds and land earn the same return, positive infinite debt rollover is possible exactly when land contains a bubble, and issuing debt only replaces part of that bubble without changing the allocation. The rollover question therefore reduces to whether land prices contain a bubble.

Theorem~\ref{thm:positive_land_share} gives a condition under which land cannot contain a bubble and debt cannot be rolled over forever. If land's income share remains bounded away from zero, every equilibrium is fundamental and Pareto efficient, and positive infinite debt rollover is impossible. The argument is as follows. A positive land share makes rents comparable to aggregate output. Because land has a finite price, aggregate resources must then have finite present value. It follows from the standard proof of the first welfare theorem that the equilibrium is Pareto efficient, and date-0 budget constraints imply that the land price equals the present value of rents. The debt--bubble equivalence then rules out rollover. This result requires neither a particular preference specification nor convergence to a balanced growth path. What matters is whether land's income share stays bounded away from zero, not merely whether land remains productive. Therefore, if rollover is possible, the land share must become arbitrarily close to zero.

We then specialize the model to a two-period OLG economy. Let $G$ denote the asymptotic growth rate of aggregate income, $G_r$ the asymptotic growth rate of land rent, and $R^{\mathrm a}$ the autarky gross interest rate---the limiting intertemporal marginal rate of substitution if asset trade became negligible. Under the condition $R^{\mathrm a}<G_r<G$, Proposition~\ref{prop:bubble_necessity} shows that an equilibrium exists and every equilibrium is asymptotically bubbly: the land price does not vanish relative to aggregate income even though rent does. If the land price became negligible relative to aggregate income, the allocation would approach autarky and the interest rate would approach $R^{\mathrm a}$. Because rent grows faster than $R^{\mathrm a}$, the present value of future rents would then be infinite, contradicting the finite price of land. This result follows directly from the Bubble Necessity Theorem of \citet{HiranoToda2025JPE}; allowing cohort size to vary requires only aggregating within each generation. Our contribution is Theorem~\ref{thm:necessity}: under an additional regularity condition on preference curvature, every equilibrium is Pareto efficient. Lemma~\ref{lem:debt_bubble} then allows the government to replace any constant fraction of the land bubble with debt and roll the debt over forever. Rollover need not correct overaccumulation: it may instead replace part of a private bubble in a Pareto-efficient equilibrium without changing the allocation.

We then apply the general theory to an economy with CES production. When aggregate labor expands, Cobb--Douglas production assigns land a constant share and precludes rollover. If the elasticity of substitution between labor and land exceeds one, the land share vanishes. This condition is necessary but not sufficient for rollover: our bubble-necessity result also requires the autarky interest rate to lie below rent growth. The model allows population, individual labor endowments, and Hicks-neutral productivity to have distinct asymptotic growth rates. Whether rollover is possible therefore depends on both production, which determines the asymptotic importance and rent growth of land, and households' demand for saving, which determines the counterfactual interest rate.

Finally, we provide an empirical illustration using Iowa farmland. In an asymptotically bubbly equilibrium, rent becomes small relative to the economy while the land price does not, so the cash yield falls. The Iowa data show this pattern. From 1994 to 2025, cash rent relative to US GDP fell by 37 percent, whereas land value relative to GDP doubled. As a result, the statewide cash yield fell from 7.5 to 2.3 percent. Because lower interest rates or higher expected rent growth can also raise land prices relative to rents, we then estimate a Gordon model in which a latent common capitalization rate absorbs all movements in the average price--rent ratio, so that the test for a bubble rests on the cross-district pattern of prices. A fall in the capitalization rate raises high price--rent ratios proportionally more, whereas a common bubble component raises low ratios proportionally more. From the mid-1990s, Iowa land prices converged across districts while the average price--rent ratio tripled. The model without a bubble cannot produce this pattern: its fit errors drift upward in the districts with low price--rent ratios and downward in those with high ratios, and the bootstrap rejects it. In contrast, the bootstrap does not reject a bubble that appears around 2002 and then earns the return on land, as in our deterministic theory. This bubble accounts for 50 percent of the average land price in 2025. These results are consistent with a land bubble, provided that district differences in fundamentals are stable over time and that the bubble is equal per unit of rent across districts.

\subsection{Related literature}

Our paper is related to four strands of literature. The first is an older literature on land and dynamic efficiency. \citet{McCallum1987} studies whether productive land eliminates capital overaccumulation in an OLG economy with reproducible capital. He assumes convergence to a steady state, an assumption he acknowledges comes close to requiring Cobb--Douglas production with a fixed supply of land. \citet{Homburg1991} and \citet{Rhee1991} instead show that a nonvanishing land share implies dynamic efficiency in two-period OLG economies. \citet{Kocherlakota1992} shows under short-sale constraints that bubbles require infinite-valued aggregate endowments; \citet{SantosWoodford1997} generalize this result, while \citet[p.~297]{Geanakoplos2005} establishes efficiency and fundamental land valuation in a general OLG economy. We show that a land share bounded away from zero is sufficient for aggregate resources to have finite present value without imposing a particular demographic structure. We then establish fundamental land valuation and Pareto efficiency and rule out debt rollover.

Second, our paper connects the literature on land to recent work on public debt cited at the beginning of the introduction. \citet{Homburg2014} argues that land removes the overaccumulation rationale for large and lasting public deficits, whereas recent rollover models typically abstract from productive land. We do not claim that the conclusions of those models are incorrect. Rather, the conclusions apply either to economies without productive land or to economies in which land becomes economically negligible. Once land is introduced, one must first ask whether the land price is fundamental or contains a bubble that government debt can replace.

Third, our positive result builds directly on the literature on land bubbles. \citet{Rhee1991} observes that land can become economically inessential while technologically productive and appeals to \citet{Tirole1985} to characterize the resulting equilibria. \citet{PhamToda2026ECMA} show that Tirole's bubble-necessity result for a dividend-paying asset is incorrect as stated and prove bubble necessity under additional restrictions. Proposition~\ref{prop:bubble_necessity} applies the Bubble Necessity Theorem of \citet{HiranoToda2025JPE}, though we allow cohort size to vary by aggregating within each generation. Our marginal contribution here is to prove Pareto efficiency under additional regularity conditions. The fiscal and welfare questions are therefore distinct: rollover depends on the existence of a land bubble, not on dynamic inefficiency. The Land Overvaluation Theorem of \citet{HiranoToda2025PNAS} uses a related production mechanism to explain why land becomes overvalued; we use the mechanism to study Pareto efficiency and public debt.

Finally, our empirical illustration relates to the empirical literature on farmland valuation and rational bubbles. Using data that end in the 1980s, \citet{Burt1986} finds that a dynamic model based on rents accounts for much of the movement in Illinois farmland prices, whereas \citet{Falk1991} rejects the constant-return present-value restrictions for Iowa and finds that prices overreact to rents. Focusing on the 1970s boom and 1980s bust, \citet{JustMiranowski1993} argue that inflation and the opportunity cost of capital account for much of this apparent overreaction. Using annual Iowa data for 1922--1994, \citet{FalkLee1998} use a structural time-series model to distinguish shocks that affect rents or interest rates from other price shocks. They find that the latter mainly generate transitory fads rather than rational bubbles, while permanent shocks to fundamentals explain most long-run price movements. More recently, \citet{BakerBoehljeLangemeier2014} document historically high farmland price--rent ratios and their negative association with subsequent long-horizon returns.

Within this broader literature, a small set of papers uses comparison assets to identify bubbles without fully specifying fundamental values. \citet{GiglioMaggioriStroebel2016} compare freeholds with otherwise similar ultra-long leaseholds and find no violation of the no-bubble condition in the United Kingdom and Singapore. \citet{DomeijEllingsen2020} argue that statutory enfranchisement may leave the UK comparison with little power to detect a bubble; their critique does not apply to Singapore. \citet{Tolhurst2021} generalizes the comparison-based approach to a difference-in-differences design for Midwestern farmland. He identifies a bubble from faster price growth in counties with low urban pressure than in counties with high urban pressure, assuming that their fundamental values would otherwise follow parallel trends, but his alternative includes both rational and irrational bubbles. These papers identify bubbles using designated comparison assets. We instead derive cross-district restrictions from the decomposition of land prices into fundamental and bubble components used in the theory: a fall in the common capitalization rate and the growth of a common bubble component move the dispersion of price--rent ratios in opposite directions. Our restricted specification further asks whether the inferred bubble earns the return on land, as deterministic rational-bubble theory requires.

\section{Environment and debt rollover}\label{sec:model}

We consider a deterministic economy with a single consumption good at each date $t\in \cT\coloneqq \set{0,1,\dotsc}$. There is no aggregate uncertainty, and markets are frictionless and dynamically complete. A productive, non-reproducible, and non-depreciating asset, which we call land, is available in unit supply. We abstract from capital accumulation. The familiar two-period overlapping-generations (OLG) model is a special case; we introduce that specialization formally in \S\ref{sec:rollover}.

\subsection{Agents and factor endowments}

Let $I$ be a finite or countably infinite set of agents. Each agent $i\in I$ lives for a consecutive set of dates $\cT_i\subset \cT$, which may be finite or infinite. We denote the set of agents alive at date $t$ by $I_t\coloneqq \set{i\in I: t\in \cT_i}$. Agent $i$ consumes $c_{it}\ge0$ and owns $h_{it}\ge0$ efficiency units of labor at $t\in\cT_i$. Agents supply labor inelastically, and firms treat all labor services as perfect substitutes. Aggregate labor input is $H_t\coloneqq\sum_{i\in I_t}h_{it}>0$. We assume that aggregate labor and consumption are finite at every date. Each agent has a locally nonsatiated preference relation over consumption during the agent's lifetime, where we use the weighted $\ell^1$-norm to define the topology (see Appendix \ref{subsec:finite_value} for details). We impose no further restriction on preferences in this section.

\subsection{Land and production}

A representative competitive firm produces the consumption good using labor $H$ and land $X$ with a neoclassical production technology.

\begin{asmp}[Production technology]\label{asmp:production}
For every $t$, the production function $F_t:\R_+^2\to\R_+$ is increasing, concave, homogeneous of degree one, and continuously differentiable on $\R_{++}^2$, with $F_{t,H}>0$ and $F_{t,X}\ge 0$. Furthermore, $F_{t,X}(H_t,1)>0$ infinitely often.
\end{asmp}

Assumption \ref{asmp:production} does not impose strict concavity, so the environment includes linear technologies as a special case. The condition $F_{t,X}(H_t,1)>0$ infinitely often implies that land remains productive.

Firms employ all available labor and land, so output is $Y_t=F_t(H_t,1)$. We denote the wage per efficiency unit of labor and the land rent by $w_t$ and $r_t$. Profit maximization implies $w_t=F_{t,H}(H_t,1)$ and $r_t=F_{t,X}(H_t,1)$.
By constant returns to scale and Euler's theorem, we have $Y_t=w_tH_t+r_t$. We denote agent $i$'s non-asset income by $e_{it}\coloneqq w_th_{it}$. Therefore,
\begin{equation}
    E_t\coloneqq\sum_{i\in I_t}e_{it}=Y_t-r_t. \label{eq:aggregate_non-asset_income}
\end{equation}
Although we call $e_{it}$ labor income, it may include taxes and transfers. For example, a balanced pay-as-you-go Social Security system redistributes non-asset income across agents without changing $E_t$.

Let $x_{i,-1}$ denote agent $i$'s landholding at the opening of date $0$. Only agents alive at date $0$ may hold land initially: we take $x_{i,-1}\ge0$ for $i\in I_0$, set $x_{i,-1}=0$ for $i\notin I_0$, and impose the unit-supply condition $\sum_{i\in I_0}x_{i,-1}=1$. Production and factor payments precede asset trading at every date. Thus, an agent who enters a date owning land receives that date's rent and may then sell the land at its ex-dividend price.

Because factor supplies are inelastic and production is static, technology and factor endowments determine $\set{(e_{it},r_t)}$ independently of household saving. Consider an exchange economy in which agent $i$ receives goods endowment $e_{it}=w_th_{it}$ and a long-lived asset pays dividend $r_t$. By \eqref{eq:aggregate_non-asset_income}, its aggregate resources are $E_t+r_t=Y_t$, so it has the same household budget sets and aggregate resource constraints as the production economy. Conversely, consider an exchange economy with nonnegative endowments $e_{it}$ and dividends $r_t$ such that aggregate endowment is positive and finite at every date and $r_t>0$ infinitely often. Set
\begin{equation}
    F_t(H,X)=H+r_tX,
    \qquad
    h_{it}=e_{it}. \label{eq:linear_representation}
\end{equation}
Then $F_{t,H}=1$ and $F_{t,X}=r_t$, so $w_t=1$, agent $i$ receives labor income $h_{it}=e_{it}$, and land pays dividend $r_t$, so Assumption \ref{asmp:production} holds. We call two economies \emph{equivalent} if their competitive equilibria correspond one to one and have the same consumption allocations, portfolios, land prices, and interest rates. This equivalence allows us to work with the simpler endowment economy without loss of generality and then translate the conclusions back to the production economy. We record this reduction in the following lemma.

\begin{lem}[Production--endowment equivalence]\label{lem:production_endowment}
The production economy is equivalent to its induced endowment economy. Conversely, every endowment economy described above is equivalent to the linear-production economy in \eqref{eq:linear_representation}.
\end{lem}

Nonlinear production functions restrict the jointly feasible paths of non-asset income and land rents and generate the cases of interest below.

\subsection{Sequential equilibrium}

Agents trade sequentially rather than trading claims to all future goods at date $0$. We denote the ex-dividend land price by $P_t>0$ and agent $i$'s landholdings after date-$t$ trading by $x_{it}$. Agents may take positive or negative land positions, subject to the solvency conditions below.\footnote{We allow short sales of land to keep the asset structure minimal. If one instead wishes to impose $x_{it}\ge0$, one may introduce zero-net-supply one-period private claims to permit borrowing and lending. Under certainty and in the absence of security-specific portfolio constraints, the claims and land have the same return and can be consolidated, so the alternative formulation generates the same allocations.} We assume that at least one agent can carry land between every pair of adjacent dates. Because $P_t>0$, land then spans one-period intertemporal transfers under certainty.

Define agent $i$'s birth and death dates by $\ubar{t}_i\coloneqq\min\cT_i$ and $\bar{t}_i\coloneqq\sup\cT_i$, where $\bar{t}_i=\infty$ if the agent is infinitely lived. If $\ubar{t}_i>0$, we impose $x_{i,\ubar{t}_i-1}=0$. Thus, agents born after date $0$ enter the economy without an inherited asset position and hence without inherited wealth or debt. This is the natural counterpart of the usual OLG assumption that each new generation starts life with only its non-asset endowment. At every date $t\in\cT_i$, agent $i$ faces the budget constraint
\begin{equation}
    c_{it}+P_tx_{it}
    =e_{it}+(P_t+r_t)x_{i,t-1}. \label{eq:sequential_budget}
\end{equation}
If $\bar{t}_i<\infty$, terminal solvency requires $P_{\bar{t}_i}x_{i,\bar{t}_i}\ge0$. Local nonsatiation and the absence of a bequest motive make this constraint bind; without loss of generality, the agent liquidates the position at $\bar{t}_i$. An infinitely lived agent faces the no-Ponzi condition
\begin{equation}
    \liminf_{T\to\infty}q_TP_Tx_{iT}\ge0. \label{eq:no_Ponzi}
\end{equation}
where the positive sequence $\set{q_t}$ is defined below. Condition \eqref{eq:no_Ponzi} restricts the marked-to-market value of the agent's position. It prevents an agent from financing consumption by perpetually rolling over a short position, while allowing the discounted value of a positive land position ($q_TP_T$) to remain positive.

Because the economy is deterministic, the return on land is riskless. No arbitrage therefore equates it to the gross equilibrium interest rate:
\begin{equation}
    R_t=\frac{P_{t+1}+r_{t+1}}{P_t}. \label{eq:Rt1}
\end{equation}
Because agents may freely adjust their land positions, land dynamically completes intertemporal trade in this deterministic economy. We define $q_t$ as the date-$0$ price of a zero-coupon bond that pays one unit of the consumption good at date $t$. Specifically, $q_0=1$ and $q_t\coloneqq 1/(R_0\dotsb R_{t-1})$ for $t\ge 1$.

A competitive equilibrium without government debt consists of consumption and landholding plans and sequences of factor prices, land prices, and interest rates such that agents optimize subject to their sequential budget and solvency constraints, firms maximize profits, and markets clear. Market clearing requires $\sum_{i\in I_t}c_{it}=Y_t$ and $\sum_{i\in I_t}x_{it}=1$. When $I_t$ is countably infinite, we require landholdings to be absolutely summable at every date, and impose the same condition on government bond holdings when they are present.

\subsection{Land bubbles and government debt}

The fundamental value of land is the present value of future rents,
\begin{equation*}
    V_t\coloneqq\frac{1}{q_t}\sum_{s=t+1}^\infty q_sr_s.
\end{equation*}
Iterating \eqref{eq:Rt1} implies $P_t\ge V_t$ (Appendix~\ref{sec:bubble}). We define the land-bubble component by $B_t\coloneqq P_t-V_t\ge0$.
An equilibrium is \emph{fundamental} if $B_t=0$ and \emph{bubbly} if $B_t>0$. Equation \eqref{eq:Rt1} and the definitions of $q_t$ and $V_t$ imply, respectively,
\begin{equation*}
    P_t=\frac{P_{t+1}+r_{t+1}}{R_t},
    \qquad
    V_t=\frac{V_{t+1}+r_{t+1}}{R_t}.
\end{equation*}
Subtracting the second equation from the first shows that the bubble component satisfies
\begin{equation}
    B_{t+1}=R_tB_t. \label{eq:bubble_growth}
\end{equation}

To study government debt rollover, we next introduce one-period government bonds. Let $d_{it}$ be the market value of government bonds held by agent $i$ after date-$t$ trading and let $D_t=\sum_i d_{it}$ be their aggregate market value. Agents may take positive or negative bond positions, and we set $d_{i,\ubar{t}_i-1}=0$. Thus, as in the no-debt economy, there are no security-specific portfolio constraints. We denote the land price in an equilibrium with government debt by $\widetilde P_t$. No arbitrage requires the bond return to equal the land return, $R_t=(\widetilde P_{t+1}+r_{t+1})/\widetilde P_t$, so the bond position $d_{it}$ at date $t$ pays $R_td_{it}$ at date $t+1$. In the sequential budget constraints, total asset purchases are $\widetilde P_tx_{it}+d_{it}$, and the preceding portfolio pays $(\widetilde P_t+r_t)x_{i,t-1}+R_{t-1}d_{i,t-1}$. The solvency condition applies to the total portfolio value $\widetilde P_tx_{it}+d_{it}$. The government enters date $0$ with no outstanding liabilities and chooses an initial debt issuance $D_0\ge0$. To preserve the distribution of initial wealth, agent $i\in I_0$ receives the lump-sum transfer $D_0x_{i,-1}$, which is proportional to initial land holdings. At every subsequent date, the government runs a zero primary balance and issues new debt to repay the old debt. Hence
\begin{equation}
    D_{t+1}=R_tD_t. \label{eq:rollover}
\end{equation}
We say that \emph{positive infinite debt rollover} occurs if $D_0>0$, equivalently if $D_t>0$ for every $t$.

Consider an equilibrium with positive infinite debt rollover and define the consolidated asset price $P_t\coloneqq\widetilde P_t+D_t$. For every agent, define consolidated landholdings by $z_{it}\coloneqq (\widetilde P_tx_{it}+d_{it})/P_t$. Land and bond market clearing imply $\sum_i z_{it}=1$. Moreover, no arbitrage and \eqref{eq:rollover} imply that the acquisition cost and subsequent payoff of the consolidated position are
\begin{equation*}
    P_tz_{it}=\widetilde P_tx_{it}+d_{it},
    \qquad
    (P_{t+1}+r_{t+1})z_{it}
    =R_t(\widetilde P_tx_{it}+d_{it}).
\end{equation*}
Replacing land and government bonds with the consolidated position therefore preserves every agent's sequential budget set and no-Ponzi condition. At date $0$, the increase $D_0x_{i,-1}$ in the value of agent $i$'s initial land exactly replaces the government transfer. Removing government debt and replacing $\widetilde P_t$ with $P_t$ consequently yields an equilibrium without government debt with the same consumption allocation and interest rates.

Let $B_t=P_t-V_t$ be the bubble component in this no-debt equilibrium. Because $\widetilde P_t\ge V_t$, we have $B_t=(\widetilde P_t-V_t)+D_t\ge D_t>0$. Furthermore, \eqref{eq:bubble_growth} and \eqref{eq:rollover} imply that $B_t$ and $D_t$ grow at the same rate. Therefore, $D_t=\theta B_t$ for all $t$, where $\theta\coloneqq D_0/B_0\in(0,1]$.

Conversely, consider a bubbly equilibrium without government debt. Choose an initial debt issuance $D_0\in(0,B_0]$, set $\theta=D_0/B_0$, and define $D_t\coloneqq\theta B_t$ and $\widetilde P_t\coloneqq P_t-D_t$. Equation \eqref{eq:bubble_growth} implies \eqref{eq:rollover}. If $z_{it}$ denotes agent $i$'s landholdings in the no-debt equilibrium, set $x_{it}=z_{it}$ and $d_{it}=D_tz_{it}$. The two positions have total acquisition cost $P_tz_{it}$ and the same subsequent payoff as the original land position. At date $0$, the government transfers $D_0x_{i,-1}$ to each initial landowner. Reversing the consolidation argument yields an equilibrium with government debt that has the same consumption allocation and interest-rate sequence.

We thus obtain the following lemma.

\begin{lem}[Debt--bubble equivalence]\label{lem:debt_bubble}
Suppose the proceeds of an initial debt issuance are distributed among initial landowners in proportion to their land endowments. An equilibrium with positive infinite debt rollover exists if and only if the corresponding economy without government debt has a bubbly equilibrium. Given a bubbly equilibrium with initial bubble $B_0$, every issuance $D_0\in(0,B_0]$ is feasible. The corresponding equilibria have the same consumption allocation and interest-rate sequence.
\end{lem}

Lemma~\ref{lem:debt_bubble} adapts the bubble-substitution arguments of \citet[\S\S5 and 7]{Tirole1985}, \citet{OConnellZeldes1988}, and \citet{Kocherlakota2013} to our environment. Henceforth, we study equilibria without government debt and invoke Lemma~\ref{lem:debt_bubble} to discuss rollover. Two qualifications matter for the fiscal interpretation. First, the lemma concerns only whether debt can be rolled over. It permits the government to replace at most $B_0$ of a preexisting private bubble while leaving the allocation unchanged. Using the proceeds for government spending would instead change the allocation. Second, rollover does not diagnose dynamic inefficiency. It is feasible precisely when land prices contain a bubble that public debt can replace. The next two sections characterize when such a bubble can and cannot exist.

\section{When land precludes debt rollover}\label{sec:benchmark}

The conventional argument for infinite debt rollover starts from a landless economy in which the interest rate may fall below economic growth. We show that productive land changes the argument: if its rents remain a nonnegligible share of output, the economy can support neither a land bubble nor government debt rolled over forever.

\subsection{General land-share criterion}\label{subsec:positive_land_share}

Define land's share of output at date $t$ by
\begin{equation*}
    s_t^X\coloneqq\frac{r_t}{Y_t}
    =\frac{F_{t,X}(H_t,1)}{F_t(H_t,1)}.
\end{equation*}
The following theorem, which is our first main result, gives a sufficient condition for land to rule out infinite debt rollover.

\begin{thm}[Land-share criterion]\label{thm:positive_land_share}
Suppose
\begin{equation}
    \liminf_{t\to\infty}s_t^X>0. \label{eq:positive_land_share}
\end{equation}
Then every equilibrium is fundamental and Pareto efficient. Consequently, positive infinite debt rollover is impossible.
\end{thm}

The force of Theorem~\ref{thm:positive_land_share} is not merely that a positive land share implies dynamic efficiency. The theorem links four objects that the literature has typically studied separately:
\begin{enumerate*}
	\item a positive land share makes rents comparable in size to aggregate output;
	\item because the land price is finite, rents have finite present value, and hence so do aggregate resources;
	\item finite-valued resources imply both that land is fundamental and that the equilibrium allocation is Pareto efficient; and
	\item the debt--bubble equivalence (Lemma~\ref{lem:debt_bubble}) converts the absence of a land bubble into the impossibility of positive infinite debt rollover.
\end{enumerate*}
Thus, one observable production-side condition yields fundamental valuation, Pareto efficiency, and a fiscal impossibility result. Earlier papers establish parts of this logical chain in more specialized environments \citep{Homburg1991,Rhee1991,SantosWoodford1997,Geanakoplos2005}.

Theorem~\ref{thm:positive_land_share} is unusually broad because the logical chain does not depend on a particular demographic structure or preference specification. Agents may have arbitrary finite or infinite lifetimes. Apart from local nonsatiation, preferences need not satisfy curvature, separability, discounting, or functional-form restrictions. The proof also requires neither convergence nor a comparison between one-period interest ($R$) and growth rates ($G$). The result is a sharp benchmark for deterministic, dynamically complete, frictionless markets. 

\subsection{Cobb--Douglas and the classical argument}\label{subsec:CD}

As an illustration of Theorem~\ref{thm:positive_land_share}, consider the Cobb--Douglas production function
\begin{equation}
    F_t(H,X)=A_tH^\alpha X^{1-\alpha}, \label{eq:CD}
\end{equation}
where $A_t>0$ is total factor productivity and $\alpha\in (0,1)$. Because we normalize the land supply to one, competitive factor pricing gives
\begin{equation*}
	r_t=F_{t,X}(H_t,1)=(1-\alpha)A_tH_t^\alpha=(1-\alpha)Y_t.
\end{equation*}
Thus, land always receives the constant output share $1-\alpha>0$, regardless of the paths of technology, population, and individual labor endowments. Theorem~\ref{thm:positive_land_share} immediately yields the following result.

\begin{cor}[Cobb--Douglas benchmark]\label{cor:CD}
Under \eqref{eq:CD}, every equilibrium is fundamental and Pareto efficient, and positive infinite debt rollover is impossible.
\end{cor}

The familiar comparison between the interest rate and the growth rate follows immediately on a balanced growth path. Suppose output and the land price grow at the common gross rate $G$, so $Y_{t+1}=GY_t$ and $P_{t+1}=GP_t$. The return on land then satisfies
\begin{equation*}
    R_t=\frac{P_{t+1}+r_{t+1}}{P_t}
    =G\left(1+\frac{r_{t+1}}{P_{t+1}}\right)>G.
\end{equation*}
The capital gain alone equals the economic growth rate, and the positive rental yield makes the total return strictly larger. Government debt that satisfies $D_{t+1}=R_tD_t$ would grow faster than the economy and therefore be unsustainable.

The balanced-growth argument clarifies the role of the assumptions in \citet{McCallum1987}. He studies whether productive land eliminates the capital overaccumulation problem in an OLG economy with reproducible capital. However, he assumes that the economy admits a unique steady state and that its dynamics converge to that steady state \citep[p.~333]{McCallum1987}. This assumption is much stronger than the mere existence of productive land. With a fixed aggregate supply of land and a growing labor force, the land-labor ratio changes over time, so balanced growth generally requires a special production function. Indeed, \citet[Footnote~18]{McCallum1987} acknowledges that the assumed existence of a steady state comes close to requiring a Cobb--Douglas production function, and his subsequent argument uses constant factor shares along the steady-state path. Thus, McCallum's convergence assumption effectively imposes Cobb--Douglas behavior along the limiting path, which is much stronger than the mere presence of productive land.

Under that restriction, McCallum's conclusion is correct. Constant factor shares make land rent and the land price grow at the same rate as output along the limiting balanced-growth path. The capital gain on land then equals the output growth rate, while the positive rent adds a strictly positive yield. No arbitrage consequently implies $R>G$, which eliminates overaccumulation and rules out infinite debt rollover. Corollary \ref{cor:CD} establishes the same conclusion directly and, because Theorem~\ref{thm:positive_land_share} does not require convergence, shows that exact Cobb--Douglas production rules out rollover even away from a balanced-growth path.

\citet{Homburg1991} starts precisely from the difficulty that \citet{McCallum1987} assumes away, and observes that fixed land generally makes conventional steady-state analysis problematic: as labor grows, the land-labor ratio and the rent-wage ratio need not remain constant. He therefore allows population, preferences, and the production function to vary over time and studies arbitrary growth paths. Instead of assuming convergence, he assumes that land is not asymptotically irrelevant, by which he means that land's income share is uniformly bounded away from zero. He then relates the period-by-period interest and output growth factors to land's income share and the old agents' consumption share. Iterating that relation shows that the ratio of compound growth to compound interest converges to zero, which is sufficient for dynamic efficiency. \citet{Homburg1991}'s argument therefore requires neither constant interest and growth rates nor convergence to balanced growth.

Theorem~\ref{thm:positive_land_share} retains Homburg's insight while removing the two-period OLG structure and adding conclusions about asset valuation and debt rollover. In this sense, the land-share condition \eqref{eq:positive_land_share}, rather than balanced growth or the existence of a steady state, is the economically essential restriction behind the classical land argument.

\subsection{Vanishing land shares}\label{subsec:vanishing_land_share}

Theorem~\ref{thm:positive_land_share} also identifies what must fail for debt rollover to become possible. We first characterize technologies whose land share converges to a positive interior limit.

\begin{prop}[Regularly varying production]\label{prop:laborshare1}
Suppose Assumption \ref{asmp:production} holds and there exist a positive sequence $\set{A_t}$ and a function $F$ such that $F_t=A_tF$ for every $t$. Let $f(H)=F(H,1)$. Then the land share at $(H,1)$ converges to $1-\alpha\in(0,1)$ as $H\to\infty$ if and only if there exist a constant $A>0$ and a function $\varepsilon$ with $\varepsilon(H)\to0$ such that
\begin{equation}
    f(H)=AH^\alpha
    \exp\left(\int_1^H\frac{\varepsilon(s)}{s}\diff s\right). \label{eq:asymptotic_CD}
\end{equation}
\end{prop}

\begin{proof}
Because $F_t=A_tF$, the total factor productivity $A_t$ cancels from the land share. Euler's theorem therefore implies that the land share equals
\begin{equation*}
	F_X(H,1)/F(H,1)=1-Hf'(H)/f(H).
\end{equation*}
Hence, the land share converges to $1-\alpha$ if and only if $Hf'(H)/f(H)=\alpha+\varepsilon(H)$ for some $\varepsilon(H)\to0$. Dividing by $H$, integrating from $1$ to $H$, and exponentiating yields \eqref{eq:asymptotic_CD} with $A=f(1)$. Reversing these steps proves the converse.
\end{proof}

Exact Cobb--Douglas production corresponds to $\varepsilon=0$. Thus, a constant interior limiting land share imposes a Cobb--Douglas exponent up to a slowly varying factor (the exponential part in \eqref{eq:asymptotic_CD}). Theorem~\ref{thm:positive_land_share}, however, applies more broadly; for example, it also applies when the land share converges to one.

More importantly, Theorem~\ref{thm:positive_land_share} identifies a necessary departure from the standard benchmark.

\begin{cor}[Vanishing land share]\label{cor:vanishing_land_share}
If positive infinite debt rollover is possible, then $\liminf_{t\to\infty}s_t^X=0$. In particular, if the land share converges, it must converge to zero.
\end{cor}

For the constant elasticity of substitution (CES) production functions, this condition has an especially transparent interpretation. Suppose $H_t\to\infty$ and
\begin{equation}
    F_t(H,X)=A_t
    \left(\alpha H^{1-1/\sigma}
    +(1-\alpha)X^{1-1/\sigma}\right)^{\frac{1}{1-1/\sigma}}, \label{eq:CES}
\end{equation}
where $\sigma>0$ is the elasticity of substitution and $\alpha\in (0,1)$ is a parameter, and we interpret the case $\sigma=1$ as Cobb--Douglas. At $X=1$, the land share is
\begin{equation*}
    s_t^X
    =\frac{1-\alpha}
    {\alpha H_t^{1-1/\sigma}+1-\alpha}.
\end{equation*}
If $\sigma<1$, then $s_t^X\to1$; if $\sigma=1$, then $s_t^X=1-\alpha$; and if $\sigma>1$, then $s_t^X\to0$. We therefore obtain an immediate necessary condition.

\begin{cor}[CES production]\label{cor:CES}
In a CES economy with $H_t\to\infty$, the condition $\sigma>1$ is necessary for positive infinite debt rollover.
\end{cor}

\citet{Rhee1991} independently identifies the land-share boundary. His Proposition~1 shows that a nonvanishing land-income share suffices for dynamic efficiency and acknowledges the independent result of \citet{Homburg1991}. Rhee also emphasizes the important distinction between technological and economic essentiality: the positive marginal product of land does not suffice if land's income share vanishes. In his CES example without land-augmenting technical progress, $\sigma>1$ can make the land share vanish and dynamic inefficiency possible. He also stresses that strong substitutability is insufficient when technical progress prevents the land share from vanishing.

Corollary \ref{cor:CES} sharpens the implication for our fiscal question: in the growing CES economy studied here, $\sigma>1$ is necessary for infinite debt rollover. It is not sufficient, because a vanishing land share does not by itself determine the equilibrium land price or the demand for saving.

The standard Cobb--Douglas argument can fail only if land rents become arbitrarily small relative to output; if the land share converges, it must vanish. Wages, rents, and output then generally grow at different rates, so the economy leaves the balanced-growth environment underlying the conventional result. The next section combines unbalanced growth and a condition on saving motives to show when this departure makes a land bubble necessary and hence permits infinite debt rollover.

\section{When land enables debt rollover}\label{sec:rollover}

In \S\ref{sec:benchmark}, we showed that land precludes debt rollover when rents remain a nonnegligible share of output. We now consider the opposite case. Suppose land rents grow more slowly than the economy, so land becomes asymptotically negligible as a factor of production, but faster than the interest rate that would prevail if asset trade became negligible. These two growth-rate inequalities force land to retain value through a bubble. Lemma \ref{lem:debt_bubble} then converts part or all of that bubble into government debt that can be rolled over forever.

To obtain positive results, we now adopt the two-period overlapping generations (OLG) specialization illustrated throughout \S\ref{sec:model}. At date $t$, a cohort of mass $N_t>0$ enters the economy, while the cohort of mass $N_{t-1}>0$ that entered at date $t-1$ is old. At date $0$, the economy also contains a mass $N_{-1}>0$ of initial old agents. Generation $t\ge0$ has utility $U_t(c_t^y,c_{t+1}^o)$, where $c_t^y$ and $c_{t+1}^o$ denote consumption when young and old. The initial old care only about their date-$0$ consumption. We impose standard conditions on preferences.

\begin{asmp}[Preferences]\label{asmp:preferences}
For every $t$, the utility function $U_t:\R_+^2\to[-\infty,\infty)$ is continuous, quasi-concave, and continuously differentiable on $\R_{++}^2$, with positive partial derivatives. In addition, $\lim_{y\downarrow0} (U_{t,1}/U_{t,2})(y,z)=\infty$ for every $z>0$.
\end{asmp}

The last condition is the young-age Inada condition. It ensures that equilibrium young-age consumption is positive, so the household first-order condition below holds with equality.

At date $t$, each young agent owns $h_t^y>0$ efficiency units of labor and each old agent owns $h_t^o\ge0$ efficiency units. Thus, generation $t$ supplies $h_t^y$ units when young and $h_{t+1}^o$ units when old. Aggregate labor input is $H_t=N_th_t^y+N_{t-1}h_t^o$. The labor endowments determine how the economy distributes non-asset income over the life cycle. We assume $h_t^y>0$ so that the young receive positive income from which to purchase land. We may interpret these endowments more generally as claims on labor income. For example, a pay-as-you-go Social Security system or a balanced tax-transfer program can reduce $h_t^y$ and raise $h_t^o$ in population-weighted amounts that leave $H_t$ unchanged. The special case $h_t^o=0$ describes agents who receive no labor or transfer income when old.

The initial old own the unit supply of land equally. We denote the amount of land that a member of generation $t$ purchases when young by $x_t$. The budget constraints of generation $t$ are
\begin{subequations}\label{eq:budget}
\begin{align}
    &\text{Young:} & c_t^y+P_tx_t&=w_th_t^y,\\
    &\text{Old:} &c_{t+1}^o&=w_{t+1}h_{t+1}^o+(P_{t+1}+r_{t+1})x_t.
\end{align}
\end{subequations}
The initial old consume their labor income, current land rent, and the proceeds from selling land, so $c_0^o=w_0h_0^o+(P_0+r_0)/N_{-1}$. Land-market clearing requires $N_tx_t=1$. Goods-market clearing is $N_tc_t^y+N_{t-1}c_t^o=F_t(H_t,1)$. The household first-order condition and \eqref{eq:Rt1} imply $R_t=(U_{t,1}/U_{t,2})(c_t^y,c_{t+1}^o)$. In equilibrium, \eqref{eq:budget} and land-market clearing yield
\begin{equation*}
    (c_t^y,c_{t+1}^o)
    =\left(w_th_t^y-\frac{P_t}{N_t},
    w_{t+1}h_{t+1}^o+\frac{P_{t+1}+r_{t+1}}{N_t}\right).
\end{equation*}
A standard truncation argument guarantees the existence of equilibrium \citep[Theorem~1]{HiranoToda2025JPE}.

For later use, define the induced per capita goods endowments $e_t^y\coloneqq w_th_t^y$ and $e_t^o\coloneqq w_th_t^o$. By Lemma \ref{lem:production_endowment}, conditional on $\set{(e_t^y,e_t^o,r_t)}$, the household problem is identical to that in a two-period OLG endowment economy with an asset paying dividend $r_t$. Conversely, \eqref{eq:linear_representation}, specialized to $(h_t^y,h_t^o)=(e_t^y,e_t^o)$, realizes any such endowment economy as a production economy. This specialization exactly nests the environment of \citet{HiranoToda2025JPE}.

\subsection{Bubble necessity and debt rollover}\label{subsec:bubble_necessity}

We state the asset-pricing result under weak conditions using Lemma \ref{lem:production_endowment}. Recall that $(e_t^y,e_t^o)=(w_th_t^y,w_th_t^o)$ denotes the labor or transfer income of a young and an old agent at date $t$, and $r_t$ denotes aggregate land rent. Define aggregate young- and old-age non-asset income by $a_t\coloneqq N_te_t^y>0$ and $b_t\coloneqq N_{t-1}e_t^o\ge0$. Thus, $a_t$ is also the aggregate income available to generation $t$ when young. Recall from \S\ref{sec:model} that $\widetilde U_t(C^y,C^o)=U_t(C^y/N_t,C^o/N_t)$ represents the preferences of this generation over aggregate consumption.

\begin{asmp}[Long-run aggregate growth]\label{asmp:long_run_growth}
For some $G\in(0,\infty)$ and $\omega\in[0,\infty)$, we have $a_{t+1}/a_t\to G$ and $b_t/a_t\to\omega$ as $t\to\infty$.
\end{asmp}

Assumption \ref{asmp:long_run_growth} restricts only asymptotic growth. It permits population and per capita income to grow simultaneously and allows their paths to be arbitrary for any finite number of dates. In particular, if population and per capita income have asymptotic gross growth rates $G_N$ and $G_e$, then aggregate income grows at rate $G=G_NG_e$.

For $(y,z)\in\R_{++}\times\R_+$, define the normalized forward-rate function
\begin{equation}
    f_t(y,z)
    \coloneqq (\widetilde U_{t,1}/\widetilde U_{t,2})(a_ty,a_tz). \label{eq:normalized_forward_rate}
\end{equation}

\begin{asmp}[Asymptotic forward rates]\label{asmp:forward_rates}
The function $f_t$ in \eqref{eq:normalized_forward_rate} admits a continuous extension to $\R_{++}\times\R_+$. There exists a continuous function $f:\R_{++}\times\R_+\to\R_+$ such that $f_t\to f$ uniformly on every compact subset of $\R_{++}\times\R_+$.
\end{asmp}

If asset trade became negligible, the aggregate consumption of generation $t$, normalized by $a_t$, would converge to $(1,G\omega)$. We therefore define the \emph{autarky gross interest rate} by
\begin{equation}
    R^{\mathrm a}\coloneqq f(1,G\omega). \label{eq:Ra}
\end{equation}
Unlike the equilibrium interest rate $R_t$, the rate $R^{\mathrm a}$ is a counterfactual object: it is the limiting intertemporal marginal rate of substitution without asset trade. Define the asymptotic growth rate of land rent by $G_r\coloneqq\limsup_{t\to\infty}r_t^{1/t}$. We call an equilibrium \emph{asymptotically bubbly} if it is bubbly and $\liminf_{t\to\infty} P_t/a_t>0$. Thus, an asymptotically bubbly equilibrium has a land price that remains nonnegligible relative to aggregate income. Under the maintained assumptions, we obtain the following result.

\begin{prop}[Bubble necessity]\label{prop:bubble_necessity}
Suppose Assumptions \ref{asmp:preferences}--\ref{asmp:forward_rates} hold. If
\begin{equation}
    R^{\mathrm a}<G_r<G, \label{eq:bubble_necessity}
\end{equation}
then an equilibrium exists and every equilibrium is asymptotically bubbly.
\end{prop}

Proposition \ref{prop:bubble_necessity} combines Theorems~1 and~2 of \citet{HiranoToda2025JPE}. Equilibrium existence already follows from their Theorem~1. The proposition formally extends their Theorem~2 by permitting arbitrary cohort sizes under conditions on aggregate income rather than imposing their unit-population formulation. This extension does not require a new asset-pricing argument. Aggregating each generation simply relabels population growth and per capita income growth as growth in $a_t$. The proof in \S\ref{subsec:proof_bubble_necessity} makes this correspondence explicit.

The economic mechanism is the same as in \citet{HiranoToda2025JPE}. If the land price became asymptotically negligible relative to aggregate income, asset trade would approach zero along a subsequence and the normalized consumption allocation would approach autarky. The equilibrium return would then approach $R^{\mathrm a}$ in \eqref{eq:Ra} along that subsequence. However, $R^{\mathrm a}<G_r$ is incompatible with a finite present value of rents. The economy resolves this contradiction through a bubble that keeps the land price nonnegligible relative to aggregate income. Because $G_r<G$, rent nevertheless becomes negligible relative to the economy.

Proposition \ref{prop:bubble_necessity} and Lemma \ref{lem:debt_bubble} already imply that positive infinite debt rollover is feasible under the weak conditions above. To determine the welfare properties of these equilibria, we now impose additional structure.

\begin{asmp}[Welfare regularity]\label{asmp:welfare}
The following conditions hold.
\begin{enumerate}
\item There exist a positive sequence $\set{\lambda_t}$, constants $G_N,G_e\in(0,\infty)$, limits $e^y>0$ and $e^o\ge0$ such that $N_{t+1}/N_t\to G_N$, $\lambda_{t+1}/\lambda_t\to G_e$, and $(e_t^y,e_t^o)/\lambda_t\to(e^y,e^o)$.
\item Preferences take the calendar-date-normalized form
\begin{equation}
    U_t(c^y,c^o)
    =U(c^y/\lambda_t,c^o/\lambda_{t+1}), \label{eq:normalized_preferences}
\end{equation}
where $U:\R_+^2\to [-\infty,\infty)$ is continuous, quasi-concave, twice continuously differentiable on $\R_{++}^2$, differentiably strictly increasing, and differentiably strictly quasi-concave.
\item The forward-rate function $\cR\coloneqq U_1/U_2$ admits a continuous extension to $\R_{++}\times\R_+$. For every compact interval $J\subset\R_{++}$, we have
\begin{equation*}
    \lim_{y\downarrow0}\inf_{z\in J}\cR(y,z)=\infty.
\end{equation*}
\end{enumerate}
\end{asmp}

Assumption \ref{asmp:welfare} makes preferences stationary in efficiency units but does not require homotheticity. It also allows population and per capita income to grow simultaneously and implies Assumption \ref{asmp:preferences}. Together with Assumption \ref{asmp:long_run_growth}, it implies $G=G_NG_e$ and $\omega=e^o/(G_Ne^y)$. Moreover, the normalized forward rates converge uniformly on compact sets to $f(y,z)=G_e\cR(e^yy,e^yz/G_e)$, so Assumption \ref{asmp:forward_rates} holds automatically. Equation \eqref{eq:Ra} then gives the autarky interest rate $R^{\mathrm a}=G_e\cR(e^y,e^o)$. The continuous extension of $\cR$ to $\R_{++}\times\R_+$ accommodates the economically important case $e^o=0$. The last condition in Assumption \ref{asmp:welfare} is a uniform young-age Inada condition.

For example, consider constant relative risk aversion (CRRA) preferences
\begin{equation}
    U(c^y,c^o)=u(c^y)+\beta u(c^o),
    \qquad
    u'(c)=c^{-\gamma},\quad \beta,\gamma>0. \label{eq:CRRA_preferences}
\end{equation}
These preferences satisfy the uniform young-age Inada condition for every $\gamma>0$. The autarky interest rate is $R^{\mathrm a}=(G_e/\beta)(e^o/e^y)^\gamma$, which remains valid for $e^o=0$. We now obtain our second main result.

\begin{thm}[Pareto efficiency and infinite debt rollover]\label{thm:necessity}
Suppose Assumptions \ref{asmp:long_run_growth} and \ref{asmp:welfare} and the bubble necessity condition \eqref{eq:bubble_necessity} hold. Then an equilibrium exists, and every equilibrium is asymptotically bubbly and Pareto efficient. Furthermore, for every equilibrium and every initial debt issuance $D_0\in(0,B_0]$, there exists a corresponding equilibrium with positive infinite debt rollover that has the same Pareto-efficient allocation and interest-rate sequence.
\end{thm}

Theorem~\ref{thm:necessity} contains the paper's welfare contribution. Bubble necessity alone does not imply Pareto efficiency. The intuition is as follows. The bubble grows at the interest rate, $B_{t+1}=R_tB_t$, yet remains comparable to aggregate income. Since $B_t\le P_t<a_t$ and $q_tB_t=B_0$, we have $q_ta_t>B_0$, so aggregate young-age income has infinite present value. Unlike in Theorem~\ref{thm:positive_land_share}, efficiency therefore cannot follow from the first welfare theorem. Instead, the economy lies on the boundary of the Cass criterion: interest rates are not persistently below growth, so no reallocation across generations makes every generation better off. \S\ref{subsec:proof_necessity} detrends the economy by aggregate population and per capita income, verifies the \citet{Cass1972} criterion, and uses the curvature in Assumption \ref{asmp:welfare} to establish efficiency. Lemma \ref{lem:debt_bubble} then allows the government to replace any constant fraction of the land bubble with public debt without changing consumption or interest rates. Infinite debt rollover therefore coexists with Pareto efficiency even though the equilibrium interest rate does not exceed the growth rate in the long run.

The distinction from the standard Samuelson--Diamond--Tirole pure-bubble mechanism matters. In that mechanism, adding a pure-bubble asset can support an efficient monetary or bubbly equilibrium, while the inefficient autarkic or nonmonetary equilibrium remains available. Here land pays positive rents, and \eqref{eq:bubble_necessity} eliminates the asymptotically bubbleless branch.

\subsection{CES production}\label{subsec:CES_rollover}

To illustrate Proposition~\ref{prop:bubble_necessity} and Theorem~\ref{thm:necessity}, we specialize to the CES technology in \eqref{eq:CES} while allowing population, individual labor endowments, and total factor productivity to grow. Let $\ell_t>0$ be a common scale for individual labor endowments. Suppose
\begin{equation}
    N_{t+1}/N_t\to G_N>0,
    \quad
    \ell_{t+1}/\ell_t\to G_H>0,
    \quad
    A_{t+1}/A_t\to G_A>0,
	\quad
    G_NG_H>1. \label{eq:CES_growth}
\end{equation}
Neither population, labor endowments, nor productivity needs to follow an exact geometric path. In particular, population may shrink provided individual labor endowments grow sufficiently quickly that aggregate labor expands. The constant-labor-endowment specification is the special case $\ell_t=1$ and $G_H=1$. Suppose the elasticity of substitution satisfies $\sigma>1$ and
\begin{equation}
    (h_t^y,h_t^o)/\ell_t
    \to(h^y,h^o),
    \qquad h^y>0,\quad h^o\ge0. \label{eq:scaled_labor_endowments}
\end{equation}
Thus, $G_H$ is the asymptotic growth rate of individual labor endowments, while $G_NG_H$ is the growth rate of aggregate labor. In particular, $H_t/(N_t\ell_t)\to \bar h\coloneqq h^y+h^o/G_N>0$. The condition $\sigma>1$ makes labor and land sufficiently substitutable for land's output share to vanish as labor grows relative to the fixed stock of land. It also accords with empirical estimates for production in land-intensive sectors.\footnote{\label{fn:ES}\citet{HiranoToda2025PNAS} discuss estimates in \citet{EppleGordonSieg2010} and \citet{AhlfeldtMcMillen2018} suggesting an elasticity of substitution greater than one between land and non-land factors in real-estate production. \citet{HiranoToda2025PNAS} also document the declining GDP share of agriculture and the declining employment share of land-intensive industries during economic development.}

Competitive factor prices satisfy
\begin{subequations}\label{eq:CES_limits}
\begin{align}
    \frac{w_t}{A_t}
    &=\alpha\left(\alpha+(1-\alpha)H_t^{1/\sigma-1}\right)^{\frac{1}{\sigma-1}}
    \to \bar w\coloneqq\alpha^{\frac{\sigma}{\sigma-1}}, \label{eq:CES_wage_limit}\\
    \frac{r_t}{A_t(N_t\ell_t)^{1/\sigma}}
    &=(1-\alpha)\left(\frac{H_t}{N_t\ell_t}\right)^{1/\sigma}
    \left(\alpha+(1-\alpha)H_t^{1/\sigma-1}\right)^{\frac{1}{\sigma-1}}
    \to \bar r\bar h^{1/\sigma},
    \label{eq:CES_rent_growth}
\end{align}
\end{subequations}
where $\bar r\coloneqq(1-\alpha)\alpha^{\frac{1}{\sigma-1}}$. Hence the induced endowments, measured in productivity and labor-endowment units, satisfy
\begin{equation*}
    \left(\frac{e_t^y}{A_t\ell_t},\frac{e_t^o}{A_t\ell_t}\right)
    \to(e^y,e^o)
    \coloneqq(\bar wh^y,\bar wh^o).
\end{equation*}
Let preferences take the calendar-date-normalized form \eqref{eq:normalized_preferences} with $\lambda_t=A_t\ell_t$, so $G_e=G_AG_H$. Thus, generation $t$ evaluates consumption in units that absorb both productivity and individual labor-endowment growth at the dates when consumption occurs.

Under these conditions, the aggregate-income growth rate, rent growth rate, and autarky interest rate defined in \eqref{eq:Ra} are
\begin{equation*}
    G=G_AG_NG_H,
    \qquad
    G_r=G_A(G_NG_H)^{1/\sigma},
    \qquad
    R^{\mathrm a}=G_AG_H\cR(e^y,e^o).
\end{equation*}
Because $G_NG_H>1$ and $\sigma>1$, we have $G_r<G$, and $R^{\mathrm a}<G_r$ if and only if
\begin{equation}
    \cR(e^y,e^o)
    <G_N^{1/\sigma}G_H^{1/\sigma-1}. \label{eq:CES_bubble_necessity}
\end{equation}
Hicks-neutral technological progress multiplies both sides of the unnormalized inequality by $G_A$ and therefore cancels. Population growth and individual labor-endowment growth enter differently because only the latter also raises the autarky interest rate through per capita income growth. We obtain the following direct application of Proposition \ref{prop:bubble_necessity} and Theorem~\ref{thm:necessity}.

\begin{prop}[CES production]\label{prop:necessity}
Suppose Assumption \ref{asmp:production} holds, production satisfies \eqref{eq:CES} with $\sigma>1$, population, productivity, and the labor-endowment scale satisfy \eqref{eq:CES_growth}, and labor endowments satisfy \eqref{eq:scaled_labor_endowments}. Suppose preferences take the form \eqref{eq:normalized_preferences} with $\lambda_t=A_t\ell_t$, where $U$ satisfies the regularity conditions in Assumption \ref{asmp:welfare}. If \eqref{eq:CES_bubble_necessity} holds, then an equilibrium exists, every equilibrium is asymptotically bubbly and Pareto efficient, and positive infinite debt rollover is possible.
\end{prop}

Proposition \ref{prop:necessity} completes the boundary identified in \S\ref{subsec:vanishing_land_share}. Under CES production, $\sigma>1$ is necessary for debt rollover because only then does the land share vanish. It is not sufficient by itself: the autarky interest rate must also lie below rent growth. This second condition depends on agents' saving motives and on the distribution of labor or transfer income over the life cycle.

The production mechanism in Proposition \ref{prop:necessity} is closely related to the Land Overvaluation Theorem of \citet{HiranoToda2025PNAS}. That paper allows a general neoclassical production function and aggregate uncertainty and shows that unbalanced factor-augmenting productivity growth, together with an elasticity of substitution above one at high input levels, makes land overvaluation necessary. However, its benchmark household side is more specialized: agents have logarithmic utility, the old receive no non-asset income, and the resulting equilibrium is unique. Its object is land overvaluation rather than welfare or fiscal policy. Proposition \ref{prop:necessity} uses the narrower deterministic CES technology but allows population and the common scale of individual labor endowments to grow at distinct rates, Hicks-neutral technological progress, general nonhomothetic preferences, and old-age labor or transfer income. Its purpose is not to present a new land-bubble mechanism; it links that production mechanism to the autarky-interest-rate condition, Pareto efficiency, and infinite debt rollover.

Under the CRRA preferences in \eqref{eq:CRRA_preferences}, the autarky interest rate is $R^{\mathrm a}=(G_AG_H/\beta)(h^o/h^y)^\gamma$. Condition \eqref{eq:CES_bubble_necessity} becomes $(1/\beta)(h^o/h^y)^\gamma<G_N^{1/\sigma}G_H^{1/\sigma-1}$. When $h^o=0$, this condition holds automatically, and the production-side growth restriction reduces to expanding aggregate labor, $G_NG_H>1$. More generally, a smaller old-age income relative to young-age income strengthens the demand for saving and lowers the autarky interest rate. Once it falls below rent growth, a fundamental land price is impossible, a land bubble necessarily arises, and that bubble creates room for perpetual debt rollover.

Proposition \ref{prop:constructed_equilibria} in Appendix~\ref{sec:equilibrium_constructions} complements this sufficient condition with two reduced-form, linear-production economies. In the first, the land price grows with the economy and the equilibrium is asymptotically bubbly and efficient. In the second, the land price grows with rent and the equilibrium is fundamental; depending on its return, it can be efficient or inefficient.

\section{Empirical illustration: Iowa farmland}\label{sec:empirical}

Recent evidence documents a large decline in real-estate yields. Using natural experiments from UK lease extensions, \citet{Backer-PeralHazellMian2026} estimate that the expected long-term housing yield fell from about 5 percent before the Great Recession to 2.7 percent in 2024. Housing therefore provides a natural setting in which to ask how such yield declines should be interpreted. It does not, however, map cleanly into our model because observed housing rent combines returns to reproducible structures and non-reproducible land. This distinction matters theoretically: \citet{Toda2025EL} shows that a stationary housing rent--price ratio can coexist with a land bubble. Applying our theory to housing would therefore require separating land and structure rents.

We instead use Iowa farmland to provide an empirical illustration of the model. Agricultural rents and values per acre provide a cleaner counterpart to the rent and price of productive land in our model. Agriculture is also a land-intensive sector whose share of output has declined markedly during development (footnote~\ref{fn:ES}), so its land is a natural place to look for the vanishing land share that the theory requires. Our cross-district tests then distinguish movements in the common capitalization rate from a bubble, rather than treating a declining cash yield alone as evidence of a bubble.

\subsection{Data}\label{subsec:empirical_data}

We obtain annual land values for Iowa, its nine crop reporting districts (CRDs), and its 99 counties from the \emph{Historical Iowa Farmland Value Survey} of Iowa State University (ISU), beginning in 1950. The US Department of Agriculture (USDA) defines CRDs to report agricultural statistics; each Iowa CRD consists of 9 to 12 neighboring counties. Henceforth, we refer to a CRD simply as a district. We obtain county and statewide cash rents for corn and soybean land from ISU's \emph{Historical Cash Rental Rates for Iowa Survey by County}, beginning in 1978. Both series are nominal dollars per acre, which we deflate by the December CPI normalized to one in 2025. Table~\ref{tab:data_sources} in Appendix~\ref{sec:farmland_data} lists the sources.

Figure~\ref{fig:regions_map} maps the nine CRDs and the ten soil regions that ISU uses to construct county values. Whereas USDA defines CRDs for statistical reporting, ISU defines soil regions to group counties with similar corn suitability ratings, which summarize how slope, erosion, drainage, and other soil characteristics affect crop yields (Appendix~\ref{subsec:data_land}). The two classifications therefore do not coincide.

\begin{figure}[!htb]
\centering
\input{fig_regions_map.tex}
\caption{Crop reporting districts and soil regions of Iowa.}\label{fig:regions_map}
\caption*{\footnotesize \emph{Notes:} Thin lines mark county borders; thick lines mark CRD borders in panel (a) and soil-region borders in panel (b). Shading distinguishes neighboring areas. Appendix~\ref{subsec:data_land} documents the soil regions. NW: Northwest; NC: North Central; NCW: North Central West; NCE: North Central East; NE: Northeast; WC: West Central; C: Central; EC: East Central; SW: Southwest; SC: South Central; SE: Southeast.}
\end{figure}

This distinction matters for our choice of units. ISU constructs county values from soil-region survey values, so counties within a soil region move together by construction, whereas district values are direct survey averages (Appendix~\ref{subsec:data_land}). We therefore use the nine districts as our units of analysis and construct each district's rent as the simple average of its county rents, matching ISU's treatment of district land values. Let $P_{it}$ and $r_{it}$ denote the real value and rent per acre in district $i=1,\dots,I$, where $I=9$, and let $P_t$ and $r_t$ denote the published statewide measures. These are simple county averages; acreage weighting hardly changes the cash yields (Appendix~\ref{subsec:data_statewide}).

The rent data are less reliable before 1994 because county coverage is sparse and the statewide rent is inconsistent with the county rents in 1989--1993 (Appendix~\ref{subsec:data_rents}). From 1994 on, coverage is nearly complete and $r_t$ is almost identical to the simple county average (Appendix~\ref{subsec:data_statewide}). Our main rent sample is therefore 1994--2025, with $9\times32=288$ district--year observations. We use 1978--2025 for robustness checks with statewide rent and 1950--2025 for tests requiring only land values. Appendix~\ref{subsec:data_rents} documents the few imputed county rents and shows that they do not affect the results.

Figure~\ref{fig:cash_yield}(a) shows that the statewide cash yield $r_t/P_t$ fell from 7.5 percent in 1994 to 2.3 percent in 2025, with a similar decline in every district. This pattern is consistent with an asymptotically bubbly equilibrium, in which the land price remains comparable with aggregate income while rent grows more slowly. Panel~(b) compares the rent and value of a fixed Iowa farmland stock with US GDP, the relevant scale because farmland is traded in a national capital market. From 1994 to 2025, $r_t/Y_t$ fell by 37 percent while $P_t/Y_t$ doubled: real rent grew by 0.7 percent per year, compared with 2.2 percent for GDP and 4.5 percent for the land value. Because Iowa farmland is not the aggregate land of the model, this evidence motivates rather than tests the condition $G_r<G$.

The longer series also captures the late-1970s land boom and the subsequent farm crisis. These aggregate facts do not establish a bubble, because lower interest rates or higher expected rent growth can also raise prices relative to rent and income. \S\ref{subsec:empirical_bubbles} therefore examines restrictions specific to our model.

\begin{figure}[!htb]
\centering
\includegraphics[width=\linewidth]{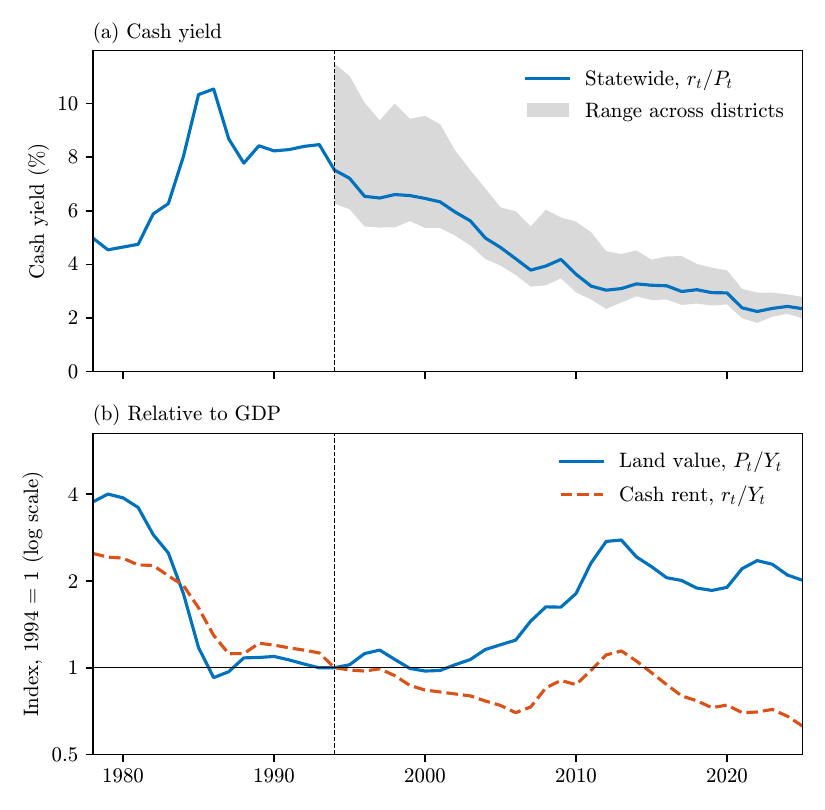}
\caption{Cash yield, rent, and land value of Iowa farmland, 1978--2025.}\label{fig:cash_yield}
\caption*{\footnotesize \emph{Notes:} In panel (a), $r_t/P_t$ is the statewide cash yield based on the land value and cash rent published by ISU. The shaded area shows the range of the cash yields $r_{it}/P_{it}$ of the nine districts, which we observe from 1994. In panel (b), $r_t$ and $P_t$ are the published statewide cash rent and land value, and $Y_t$ is US GDP deflated by the CPI, so both ratios equal the corresponding ratios of nominal values. We normalize both ratios to one in 1994. The dashed vertical line marks 1994, the first year of our main sample.}
\end{figure}

\subsection{Common factor in land rents}\label{subsec:common_rent}

Our tests in \S\ref{subsec:empirical_bubbles} require a measure of the scale of rents in each district relative to the state. Suppose that district rents have the factor structure
\begin{equation}
    r_{it}=\psi_ir_t\exp(\epsilon_{it}), \label{eq:rent_factor}
\end{equation}
where $\psi_i>0$ is a time-invariant district loading and $\epsilon_{it}$ is a small idiosyncratic error. Thus, $\psi_i$ measures the scale of rents in district $i$ relative to the statewide rent. It does not by itself measure the relative fundamental value of land, because districts may also differ in expected rent growth and in other ways that rents do not capture; \S\ref{subsec:empirical_bubbles} allows for such differences.

Taking logarithms of \eqref{eq:rent_factor}, we obtain
\begin{equation}
    \log r_{it}=\alpha_i+\beta_i\log r_t+\epsilon_{it}, \label{eq:rent_slope}
\end{equation}
where $\alpha_i=\log\psi_i$ and $\beta_i=1$. This suggests regressing $\log r_{it}$ on district fixed effects and $\log r_t$ and testing whether the common slope equals one. This test, however, has little power, because $r_t$ is itself an average of the county rents. For example, if $\log r_t$ equals the cross-district average of $\log r_{it}$ in a balanced panel, the fixed-effects estimate of the slope equals one exactly, whatever the data.\footnote{Let $\tilde y_{it}$ and $\tilde x_t$ denote $\log r_{it}$ and $\log r_t$ minus their time averages for district $i$. Then $(1/I)\sum_i\tilde y_{it}=\tilde x_t$, so the OLS estimate is $\sum_{i,t}\tilde x_t\tilde y_{it}/\sum_{i,t}\tilde x_t^2=1$. In our data, $r_t$ is an average over counties rather than districts, and the fixed-effects estimate is 1.023.} We therefore evaluate \eqref{eq:rent_factor} in two other ways. First, we impose $\beta_i=1$ and ask how much of the variation in district rents the common factor $\log r_t$ explains and whether the loadings $\psi_i$ are stable. Second, we allow $\beta_i$ to differ across districts and test whether the $\beta_i$ are equal. Because we have only nine districts but 32 years, we base inference on the time-series dimension.

We first estimate $\log\psi_i$ by the time average of $\log(r_{it}/r_t)$. With $\beta_i=1$, the common factor explains 97 percent of the variation in $\log r_{it}$ over time within districts, and the standard deviation of the error $\epsilon_{it}$ is only 0.031. The loadings are also stable: the estimates of $\log\psi_i$ from 1994--2009 and from 2010--2025 have a correlation of 0.91. Next, we estimate \eqref{eq:rent_slope} separately for each district without restricting $\beta_i$. The estimates of $\beta_i$ are heterogeneous: they range from 0.81 in East Central Iowa to 1.18 in Southwest Iowa, and a Wald test of equal slopes, which uses the joint covariance matrix of the nine regressions with Newey--West corrections,\footnote{Throughout, we use three lags for Newey--West standard errors, following the rule of thumb $\floor{4(T/100)^{2/9}}$ with $T=32$.} rejects at any conventional level ($\chi^2(8)=531$). Thus, district rents follow an approximate, but not exact, factor structure, and relative rents $r_{it}/r_t$ drift slowly over time (Appendix~\ref{subsec:district_results}).

\subsection{Cross-district implications of a land bubble}\label{subsec:empirical_bubbles}

We decompose the land price of district $i$ into its fundamental value and a bubble,
\begin{equation}
    P_{it}=V_{it}+B_{it}, \label{eq:district_decomposition}
\end{equation}
where $V_{it}$ is the present value of current and future rents and $B_{it}\coloneqq P_{it}-V_{it}\ge 0$ is the bubble. Because \eqref{eq:district_decomposition} defines $B_{it}$, it always holds. To give it content, we need a model of the fundamental value $V_{it}$ and a restriction on how the bubble differs across districts.

For the fundamental value, we use a Gordon model. Suppose that the rent of district $i$ is $\psi_ir_t$ as in \eqref{eq:rent_factor}, that it is expected to grow at rate $g_i$, and that the required return is $\rho$. With $\rho>g_i$, we have
\begin{equation}
    V_{it}=\int_0^\infty e^{-\rho s}\psi_i r_t e^{g_is}\,\diff s
    =\frac{\psi_i r_t}{\rho-g_i}. \label{eq:gordon}
\end{equation}
Thus, districts differ in their fundamental price--rent ratios $V_{it}/(\psi_ir_t)$ only through expected rent growth; Appendix~\ref{subsec:gordon_details} gives the counterpart with aggregate uncertainty. Property taxes do not change this structure, because Iowa assesses agricultural land on the basis of its productivity and net earning capacity rather than its market value.\footnote{Iowa Code \S441.21(1)(e).} Thus, the property tax is roughly proportional to rent rather than to the land price: it reduces the rent that the owner keeps but not the return that the owner requires. If the tax takes the same share of rent in every district, its effect is common across districts and can be absorbed into the common capitalization rate introduced below. More generally, districts may differ in ways that rents do not capture. For example, the share of rent paid in tax may differ across districts, land values cover all grades of farmland whereas rents cover only corn and soybean land, and nonfarm demand for land differs across districts. We therefore replace $-g_i$ in the denominator of \eqref{eq:gordon} by a general district-specific term $\delta_i$, of which differences in expected rent growth are only one example.

Differences in expected rent growth cannot account for most of $\delta_i$: the trends in relative rents $r_{it}/r_t$ range from $-0.3$ to $0.3$ percentage points per year (\S\ref{subsec:common_rent}), whereas our estimates of $\delta_i$ in model~\modelBU below range from $-1.9$ to $4.3$ percentage points. The required return and expected rent growth also change over time. We let $\rho_t$ denote the common part of the denominator at date $t$, which we call the \emph{capitalization rate}. Because a common change in $\delta_i$ can be absorbed by $\rho_t$, we normalize the cross-district average of $\delta_i$ to zero, so that $\rho_t$ is the capitalization rate of the average district. In the Gordon model \eqref{eq:gordon}, $\delta_i=\bar g-g_i$ and $\rho_t=\rho-\bar g$, where $\bar g$ is the cross-district average of $g_i$. We write the fundamental price--rent ratio as
\begin{equation}
    \frac{V_{it}}{\psi_ir_t}=\frac{1}{\rho_t+\delta_i}. \label{eq:gordon_t}
\end{equation}
Thus, at each date, investors value land as if the current $\rho_t$ will persist, and a change in $\rho_t$ is an unanticipated, permanent shock. We place no restriction on the path of $\rho_t$. This choice matters. \citet{Shiller1981} shows that stock prices move far more than the present value of dividends at a constant discount rate. Rents are also smooth whereas land prices are volatile, so the Gordon model with a constant capitalization rate, or with $\rho_t$ equal to a measured real interest rate, is easily rejected (Appendix~\ref{subsec:gordon_observed}). Such a rejection says only that the capitalization rate moves, not that land prices contain a bubble. Because $\rho_t$ is free, \eqref{eq:gordon_t} can match any path of the average price--rent ratio, and the tests below use only the cross-district pattern.

For the bubble, \eqref{eq:bubble_growth} shows that, in our deterministic theory, the bubble grows at the gross return on land. Because all districts share the required return and the property tax falls on rent rather than on land value, the same argument applies to the land in each district: both $P_{it}$ and $V_{it}$ satisfy the no-arbitrage condition with the common return $R_t$, so $B_{i,t+1}=R_tB_{it}$. Hence the ratios $B_{it}/B_{jt}$ are constant over time. The theory does not determine these ratios, because any allocation of the bubble across districts satisfies the no-arbitrage condition. We assume that the bubble is equal per unit of rent, $B_{it}=\psi_iB_t$, for two reasons. First, district-specific bubble loadings are at best weakly identified. A small change in $\rho_t$ changes the ratio \eqref{eq:gordon_t} by approximately $-1/(\rho_t+\delta_i)^2$ times the change, so a bubble with loadings proportional to $1/(\rho_t+\delta_i)^2$ is nearly indistinguishable from a change in the capitalization rate, and nine districts cannot separate the two. Second, the same assumption underlies the prediction about dispersion below.

Dividing \eqref{eq:district_decomposition} by $\psi_ir_t$, we obtain the normalized price--rent ratio
\begin{equation}
    Q_{it}\coloneqq\frac{P_{it}}{\psi_ir_t}=\frac{1}{\rho_t+\delta_i}+\frac{B_t}{r_t}. \label{eq:price_rent_model}
\end{equation}
Equation \eqref{eq:price_rent_model} implies that a bubble and a shock to the capitalization rate move districts differently. A fall in $\rho_t$ raises every price--rent ratio, but it raises the high ratios proportionally more, because $1/(\rho_t+\delta_i)$ is convex in $\rho_t$. Thus, without a bubble, the average price--rent ratio and the cross-district dispersion of $\log Q_{it}$ move in the same direction: both fall when $\rho_t$ rises and both rise when $\rho_t$ falls. In contrast, a rise in $B_t/r_t$ adds the same amount to every ratio, so it raises the low ratios proportionally more and lowers the dispersion. To examine this prediction over a long period, we use the normalized land price $X_{it}\coloneqq P_{it}/\hat\psi_i$, where $\hat\psi_i$ comes from \S\ref{subsec:common_rent}. Because $\log Q_{it}-\overline{\log Q}_t=\log X_{it}-\overline{\log X}_t$, where $\overline{\log X}_t$ denotes the cross-district average of $\log X_{it}$, the relative prices and their dispersion require only land values, and we use all years from 1950 to 2025. For this purpose, we assume that the loadings estimated from 1994--2025 rents also apply to earlier years.

\begin{figure}[!htb]
\centering
\includegraphics[width=\linewidth]{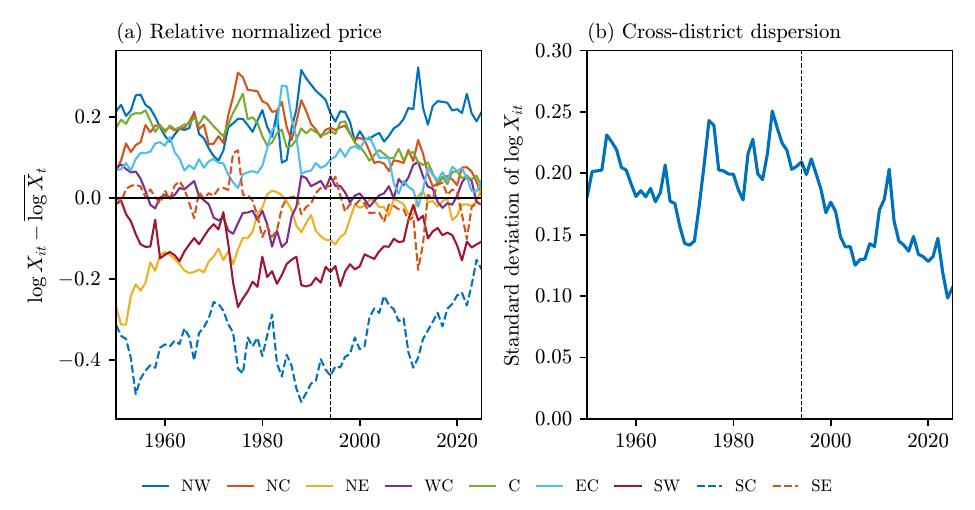}
\caption{Relative normalized land prices and their dispersion, 1950--2025.}\label{fig:relative_prices}
\caption*{\footnotesize \emph{Notes:} Panel (a) plots $\log X_{it}-\overline{\log X}_t$ for each district with $X_{it}=P_{it}/\hat\psi_i$. Panel (b) plots the cross-district standard deviation of $\log X_{it}$, which divides the sum of squared deviations by $I-1=8$. The dashed vertical line marks 1994, the first year of the district rent sample. District abbreviations are defined in Figure~\ref{fig:regions_map}.}
\end{figure}

Figure~\ref{fig:relative_prices} shows the result. Before 1994, the relative prices of some districts drift, but the dispersion fluctuates between 0.14 (in 1971) and 0.25 (in 1988) without a trend: over 1950--1993, its trend is 0.0005 per year, with standard error 0.0004. From the mid-1990s, the relative prices converge, and the dispersion falls by about half, from 0.21 in 1994 to 0.11 in 2025. Over 1994--2025, its trend is $-0.0023$ per year, with standard error 0.0005, and the dispersion in 2024 and 2025 is the lowest since 1950. The convergence is broad (Figure~\ref{fig:relative_prices}(a)), and the downward trend remains significant when we drop any single district (Appendix~\ref{subsec:district_results}). Meanwhile, the average price--rent ratio roughly triples. Without a bubble, the dispersion would have risen with the average ratio: if district yield gaps had stayed at their 1994--1999 values, the dispersion in 2025 would be about five times as large as it is (Appendix~\ref{subsec:district_results}). Convergence is instead what \eqref{eq:price_rent_model} predicts when the bubble grows relative to rent.

\subsection{Estimating the land bubble}\label{subsec:gordon_estimates}

We next estimate \eqref{eq:price_rent_model} directly over our main sample, 1994--2025. Appendix~\ref{subsec:gordon_long} reports the results for 1987--2025 and 1978--2025, which use the less reliable statewide rent before 1994. Let $R_{t\to t+s}\coloneqq R_tR_{t+1}\cdots R_{t+s-1}$ denote the cumulative gross return on land from date $t$ to $t+s$, where $R_t=(P_{t+1}+r_{t+1})/P_t$ is the statewide return \eqref{eq:Rt1} computed from the published statewide land value and cash rent. We consider three models that differ only in the bubble. Model~\modelF is fundamental: $B_t=0$ for all $t$. Model~\modelBU leaves the bubble's time path unrestricted: $B_t\ge 0$ is a free parameter in every year. Model~\modelBR assumes that the bubble emerges at $t=t_0$ and thereafter earns the land return: $B_t=0$ for $t<t_0$ and $B_t=B_{t_0}R_{t_0\to t}$ for $t\ge t_0$. Thus, model~\modelBR imposes \eqref{eq:bubble_growth}: once it appears, the bubble earns the same realized return as land. With aggregate uncertainty, this is a restriction motivated by the deterministic theory rather than an implication of no arbitrage (Appendix~\ref{subsec:gordon_details}). Because a bubble cannot appear within the model after $B_t=0$, we treat $t_0$ as a one-time unanticipated shock and estimate it together with the single bubble parameter $B_{t_0}$. The bubble grows faster than the land price because land also earns rent; their ratio grows each year at approximately the cash yield.

For each model, we estimate the parameters by nonlinear least squares on
\begin{equation}
    \log Q_{it}=\log\left(\frac{1}{\rho_t+\delta_i}+\frac{B_t}{r_t}\right)+u_{it}, \label{eq:gordon_regression}
\end{equation}
where $Q_{it}=P_{it}/(\hat\psi_ir_t)$ and $u_{it}$ is an error. To compare a restricted model with sum of squared residuals $S_0$ with a more general one with $S_1$, we use the quasi-likelihood-ratio statistic $\mathrm{LR}=IT\log(S_0/S_1)$. Its distribution is nonstandard because the errors are serially correlated, $B_t=0$ lies on a boundary, and the number of parameters grows with the sample. We therefore use a parametric bootstrap with independent AR(1) errors in each district and 999 replications. Each replication re-estimates both models and, for model~\modelBR, chooses $t_0$ over 1994--2015 by maximizing the quasi-likelihood. Appendix~\ref{subsec:gordon_details} reports a robustness check for the autoregressive coefficients.

\begin{table}[!htb]
\centering
\caption{Gordon model with and without a land bubble, 1994--2025.}\label{tab:gordon_tests}
\begin{tabular}{lccc}
\toprule
& Model~\modelF & Model~\modelBU & Model~\modelBR \\
\midrule
Bubble & $B_t=0$ & $B_t\ge 0$ free & $B_{t_0}R_{t_0\to t}$ \\
Start year $t_0$ & --- & --- & 2002 \\
Residual s.d. & 0.074 & 0.036 & 0.039 \\
\midrule
\multicolumn{4}{l}{\emph{Null hypothesis: model~\modelF}} \\
LR & --- & 420 & 369 \\
5\% critical value & --- & 85 & 47 \\
$p$-value & --- & 0.001 & 0.001 \\
\midrule
\multicolumn{4}{l}{\emph{Null hypothesis: model~\modelBR}} \\
LR & --- & 51 & --- \\
5\% critical value & --- & 68 & --- \\
$p$-value & --- & 0.17 & --- \\
\bottomrule
\end{tabular}
\caption*{\footnotesize \emph{Notes:} Nonlinear least squares estimates of \eqref{eq:gordon_regression} for the nine districts ($I=9$, $T=32$). Critical values and $p$-values come from a parametric bootstrap with 999 replications and district-specific AR(1) errors. For model~\modelBR, the data and each bootstrap replication choose $t_0$ over 1994--2015.}
\end{table}

\begin{figure}[!htb]
\centering
\includegraphics[width=\linewidth]{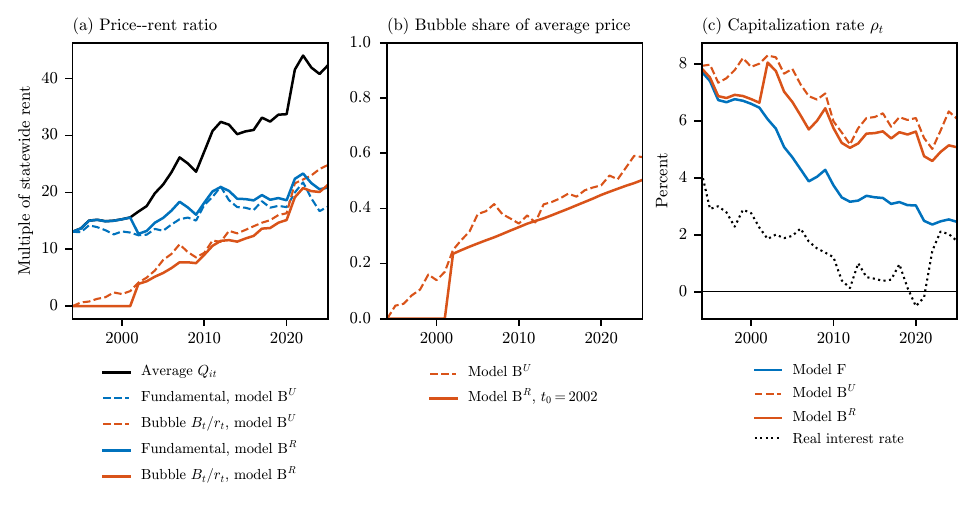}
\caption{Fundamental value and bubble of Iowa farmland, 1994--2025.}\label{fig:gordon_main}
\caption*{\footnotesize \emph{Notes:} Estimates of \eqref{eq:gordon_regression}. Panel (a) plots the cross-district average of $Q_{it}$, the average fundamental ratio $(1/I)\sum_i1/(\rho_t+\delta_i)$, and the bubble $B_t/r_t$ in models~\modelBU and~\modelBR. Panel (b) plots the bubble share $(B_t/r_t)/\bar Q_t$, where $\bar Q_t$ is the cross-district average of $Q_{it}$. Panel (c) plots the capitalization rate $\rho_t$ of the average district in models~\modelF, \modelBU, and~\modelBR and the real interest rate, measured by the 10-year Treasury yield minus the 10-year expected inflation of the Federal Reserve Bank of Cleveland, in November of each year.}
\end{figure}

Table~\ref{tab:gordon_tests} and Figure~\ref{fig:gordon_main} report the results. Model~\modelF fits poorly: its residual standard deviation is about twice that of model~\modelBU, and the bootstrap rejects it against both bubble models. Without a bubble, the capitalization rate must fall from 7.7 percent in 1994 to 2.5 percent in 2025, while the cross-district fit contradicts the convergence in Figure~\ref{fig:relative_prices}. Figure~\ref{fig:gordon_errors} shows the pattern: the errors of model~\modelF fall in districts with high initial ratios, rise in districts with low initial ratios, and drift little in the middle. Both bubble models remove these systematic trends. Their remaining errors are similar except mainly in South Central, where model~\modelBR fits worse.

\begin{figure}[!htb]
\centering
\includegraphics[width=\linewidth]{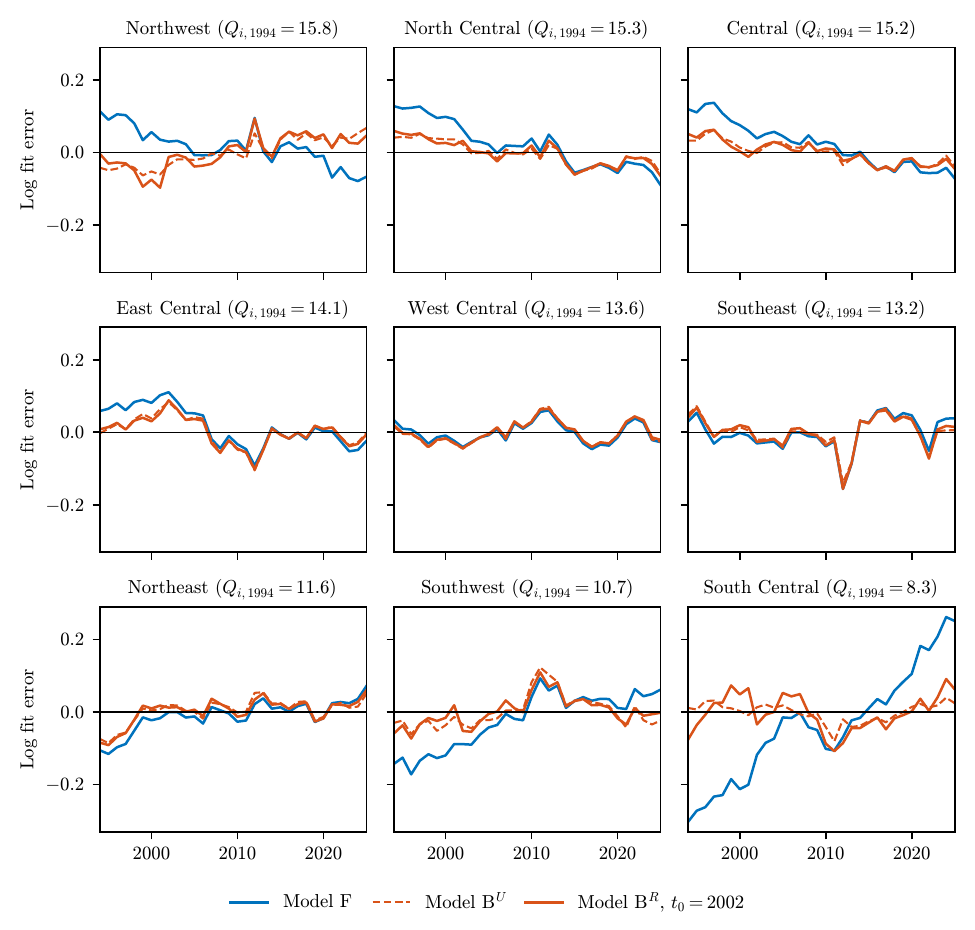}
\caption{Fit errors of the Gordon model by district, 1994--2025.}\label{fig:gordon_errors}
\caption*{\footnotesize \emph{Notes:} Each panel plots $\log Q_{it}-\log\hat Q_{it}$, where $Q_{it}=P_{it}/(\hat\psi_ir_t)$ and $\hat Q_{it}$ is the fitted value of \eqref{eq:gordon_regression} in models~\modelF, \modelBU, and~\modelBR of Table~\ref{tab:gordon_tests}. The panels are ordered by $Q_{i,1994}$, shown in each title, from the highest (top left) to the lowest (bottom right).}
\end{figure}

In model~\modelBR, the capitalization rate falls from 7.9 to 5.1 percent. Because $\rho_t$ is the required return minus expected rent growth, its level is not comparable with an interest rate, but its changes are comparable if expected rent growth and the risk premium are stable. Over 1994--2025, the real interest rate in Figure~\ref{fig:gordon_main}(c) falls by 2.4 percentage points, about as much as $\rho_t$ in models~\modelBU and~\modelBR (1.8 and 2.8 points). In contrast, model~\modelF requires a fall of 5.3 points. It therefore needs either a fall in the required return about twice as large as that in the real interest rate or a rise in expected rent growth of about 3 percentage points, whereas the real cash rent grew by only 0.7 percent per year. Model~\modelBR achieves most of the improvement of model~\modelBU with a single bubble parameter: its quasi-likelihood ratio against model~\modelF is 369, compared with 420 for model~\modelBU, which has 32 bubble parameters. A bootstrap comparison with the unrestricted model~\modelBU does not reject the restriction that the bubble earns the land return ($p=0.17$). Adjusting returns for property taxes gives similar results (Appendix~\ref{subsec:gordon_details}).

The longer samples, which use the statewide rent before 1994, give similar results (Appendix~\ref{subsec:gordon_long}). Over both 1987--2025 and 1978--2025, the bootstrap rejects model~\modelF against both bubble models at any conventional level, and model~\modelBR attributes 52 to 62 percent of the average land price in 2025 to the bubble. Over 1978--2025, however, the bootstrap also rejects model~\modelBR against model~\modelBU, because the unrestricted bubble also captures an earlier episode: it accounts for about 30 percent of the average land price around 1980 and nearly vanishes during the farm crisis of the mid-1980s. A bubble that earns the land return cannot shrink in this way, and a bubble that collapses lies outside our deterministic theory.

The quasi-likelihood selects $t_0=2002$. Because model~\modelBR makes the bubble appear at once whereas the unrestricted bubble in model~\modelBU rises gradually, we interpret $t_0$ as the approximate date by which the bubble became large, not as its literal date of birth. Appendix~\ref{subsec:gordon_details} reports the likelihood profile. In model~\modelBR, the bubble accounts for 35 percent of the average land price in 2012 and 50 percent in 2025; model~\modelBU gives 59 percent in 2025.

We interpret these results as consistent with---but not proof of---a land bubble, because they rest on specific assumptions. First, the fundamental value \eqref{eq:gordon_t} allows district differences $\delta_i$ that are constant over time. Fundamental forces that raise land values relative to rents mostly in low-rent districts would also produce convergence. Examples include nonfarm demand for recreational land and the conversion of pasture to cropland, which changes the composition of land values for all grades. Second, the bubble is equal per unit of rent; the theory does not determine its distribution across districts, and with free loadings the bubble is at best weakly identified. Third, model~\modelBR treats the appearance of the bubble as a one-time unanticipated shock, so the bubble jumps at $t_0$, whereas the unrestricted bubble in model~\modelBU rises more gradually. Moreover, the requirement that the bubble earn the realized return on land in every year is stronger than no arbitrage under aggregate uncertainty. Fourth, the land values are expert estimates rather than transaction prices. If respondents in low-value districts increasingly followed statewide price movements, the survey itself could produce convergence. However, ISU's construction of district values is unchanged since 1950 (Appendix~\ref{subsec:data_land}), and the dispersion shows no trend before 1994, so this explanation would require such behavior to strengthen only after the mid-1990s. Finally, the bootstrap conditions on the estimates $\hat\psi_i$ and assumes that the errors are independent across districts. Different respondents answer for different districts, so sampling errors are plausibly independent, but common shocks to neighboring districts would make the critical values too small.

\section{Conclusion}

This paper clarifies when productive land rules out infinite government debt rollover. If land's income share remains bounded away from zero, land prices are fundamental, equilibrium allocations are Pareto efficient, and infinite debt rollover is impossible. If the land share instead vanishes and the autarky interest rate lies below rent growth, every equilibrium contains a land bubble. Under a regularity condition, every equilibrium remains Pareto efficient, while government debt can replace part of the land bubble without changing the allocation. Asset price bubbles and infinite debt rollover therefore do not require dynamic inefficiency.

The Iowa evidence shows how the theory can be brought to data. Since the mid-1990s, cash rent has fallen relative to GDP while land value has risen, producing a sharp decline in the cash yield. A Gordon model with a free common capitalization rate can match this decline only with a cross-district pattern of prices that the data reject. In contrast, the bootstrap does not reject model~\modelBR, in which a bubble appears around 2002 and then earns the return on land, and this bubble accounts for about half of the average land price in 2025. Interpreting this evidence as a land bubble requires stable district differences in fundamentals and a bubble that is equal per unit of rent across districts.

Two extensions are natural. First, our analysis is deterministic. With aggregate uncertainty, government debt can be safe while land is risky, so the two are no longer perfect substitutes. \citet{HiranoToda2025PNAS} show that land overvaluation remains necessary with aggregate uncertainty, but extending the debt--bubble equivalence and the welfare results to that setting remains open. A stochastic setting would also allow bubbles that collapse, such as the episode before the farm crisis of the 1980s that our deterministic specification cannot capture (Appendix~\ref{subsec:gordon_long}). Second, it may seem natural to test the theory with housing rents and prices, but observed housing rent combines returns to land and structures. Indeed, \citet{Toda2025EL} shows that a stationary housing rent--price ratio can coexist with a land bubble. Applying the theory to housing would therefore require separating land and structure rents.

\appendix

\section{Proofs}\label{sec:proofs}

This appendix collects the proofs omitted from the main text. In the proof of Theorem~\ref{thm:necessity}, we state the exact version of the Cass criterion needed for the welfare argument; Appendix~\ref{sec:Cass} provides its proof and reviews the underlying literature.

\subsection{Proof of Theorem~\ref{thm:positive_land_share}}\label{subsec:finite_value}

Consider an equilibrium in the sequential economy of \S\ref{sec:model}. Extend each agent's consumption, non-asset income, and landholdings by zero outside the agent's lifetime, and write $c_i=(c_{it})_{t=0}^\infty$ and $e_i=(e_{it})_{t=0}^\infty$. Recall that $E_t=\sum_i e_{it}$ and that the initial land endowments satisfy $x_{i,-1}\ge0$ and $\sum_i x_{i,-1}=1$.

For an infinitely lived agent, let $\ell^1(q)$ denote the space of consumption streams with finite norm $\norm{c_i}_q\coloneqq\sum_{t=0}^\infty q_t\abs{c_{it}}$. Throughout this subsection, we interpret local nonsatiation for an infinitely lived agent relative to this norm on the positive cone of $\ell^1(q)$. The present-value functional is continuous in this norm, and under the hypothesis below every feasible consumption stream belongs to this space. 

The following result adapts the no-bubble result of \citet{Kocherlakota1992}, generalized by \citet{SantosWoodford1997} (see also \citet[\S3.4]{HiranoToda2024JME}), and the standard proof of the first welfare theorem.

\begin{prop}[Finite aggregate endowment value]\label{prop:finite_endowment2}
If aggregate non-asset income has finite present value, $\sum_{t=0}^\infty q_tE_t<\infty$, then the land price is fundamental, $P_0=\sum_{t=1}^\infty q_tr_t$, and the equilibrium allocation $(c_i)$ is Pareto efficient.
\end{prop}

\begin{proof}
Multiplying the sequential budget constraint \eqref{eq:sequential_budget} by $q_t$ and using $q_tR_{t-1}=q_{t-1}$ and \eqref{eq:Rt1}, we obtain the finite-horizon budget identity
\begin{equation}
    \sum_{t=0}^Tq_t(c_{it}-e_{it})+q_TP_Tx_{iT}
    =(P_0+r_0)x_{i,-1}. \label{eq:finite_horizon_budget}
\end{equation}
For a finitely lived agent, the terminal land position vanishes after liquidation, so the lifetime identity follows by taking $T$ beyond the agent's final date. Consider an infinitely lived agent. The finite-value hypothesis and nonnegativity imply $\sum_tq_te_{it}<\infty$. Equations \eqref{eq:finite_horizon_budget} and \eqref{eq:no_Ponzi} then bound the partial sums of discounted consumption, which are increasing. Hence $\sum_tq_tc_{it}<\infty$, and \eqref{eq:finite_horizon_budget} implies that $q_TP_Tx_{iT}$ has a limit as $T\to\infty$. The no-Ponzi condition \eqref{eq:no_Ponzi} makes this limit nonnegative. We thus obtain the lifetime budget constraint
\begin{equation}
    \sum_{t=0}^\infty q_tc_{it}
    \le\sum_{t=0}^\infty q_te_{it}+(P_0+r_0)x_{i,-1}. \label{eq:lifetime_budget}
\end{equation}
Because unrestricted land positions span one-period transfers, every consumption plan satisfying \eqref{eq:lifetime_budget} is sequentially affordable. To see this, recursively define landholdings from the sequential budget constraints \eqref{eq:sequential_budget}. Equation \eqref{eq:finite_horizon_budget} then makes the discounted terminal land value equal to lifetime wealth minus cumulative discounted consumption. The lifetime budget constraint \eqref{eq:lifetime_budget} implies the terminal solvency condition for a finitely lived agent and the no-Ponzi condition \eqref{eq:no_Ponzi} for an infinitely lived agent. Local nonsatiation and optimality therefore imply that \eqref{eq:lifetime_budget} binds. Summing the resulting equalities across agents gives
\begin{equation}
    \sum_{t=0}^\infty q_t\sum_i c_{it}
    =\sum_{t=0}^\infty q_tE_t+P_0+r_0. \label{eq:aggregate_budget_value}
\end{equation}
All terms are nonnegative, so \eqref{eq:aggregate_budget_value} remains valid when $I$ is countably infinite: the summations across agents and dates are monotone limits. Goods-market clearing and \eqref{eq:aggregate_non-asset_income} imply that the left-hand side of \eqref{eq:aggregate_budget_value} also equals $\sum_{t=0}^\infty q_t(E_t+r_t)$. Because aggregate non-asset income has finite present value and $q_0=1$, equating the two expressions and subtracting $\sum_{t=0}^\infty q_tE_t+r_0$ gives $P_0=\sum_{t=1}^\infty q_tr_t$.

We next prove that $(c_i)$ is Pareto efficient. Suppose, to the contrary, that a feasible allocation $(c_i')$ Pareto dominates $(c_i)$. Since agent $i$ weakly prefers $c_i'$ to $c_i$, optimality and local nonsatiation imply
\begin{equation}
    \sum_{t=0}^\infty q_tc_{it}'
    \ge\sum_{t=0}^\infty q_tc_{it}=\sum_{t=0}^\infty q_te_{it}+(P_0+r_0)x_{i,-1}. \label{eq:preferred_bundle_value}
\end{equation}
Indeed, if \eqref{eq:preferred_bundle_value} failed, continuity of the present-value functional and local nonsatiation would allow the agent to use the budget slack to choose a nearby bundle strictly preferred to $c_i'$, contradicting the optimality of $c_i$. Moreover, \eqref{eq:preferred_bundle_value} is strict for every agent who strictly prefers $c_i'$ to $c_i$. Thus, at least one agent satisfies a strict inequality. Summing these nonnegative present-value differences across agents and using feasibility, the finite-value hypothesis, and the fundamental-value identity, we obtain
\begin{equation*}
    \infty>\sum_{t=0}^\infty q_t(E_t+r_t)\ge \sum_{t=0}^\infty q_t\sum_i c_{it}'
    >\sum_{t=0}^\infty q_tE_t+P_0+r_0
    =\sum_{t=0}^\infty q_t(E_t+r_t),
\end{equation*}
which is a contradiction.
\end{proof}

\begin{proof}[Proof of Theorem~\ref{thm:positive_land_share}]
Take any equilibrium and let $q_t$ be the date-$0$ valuation of the date-$t$ good. Because $P_0$ is finite, the fundamental value of land satisfies $\sum_{t=1}^{\infty}q_tr_t=V_0\le P_0<\infty$. Condition \eqref{eq:positive_land_share} implies that there exist $\delta>0$ and $T$ such that $r_t\ge\delta Y_t$ for all $t\ge T$. Therefore,
\begin{equation*}
    \sum_{t=0}^{\infty}q_tY_t
    \le\sum_{t=0}^{T-1}q_tY_t
    +\frac{1}{\delta}\sum_{t=T}^{\infty}q_tr_t
    <\infty.
\end{equation*}
Because $E_t=Y_t-r_t\le Y_t$, Proposition \ref{prop:finite_endowment2} implies that the equilibrium is fundamental and Pareto efficient. Lemma \ref{lem:debt_bubble} then rules out an equilibrium with positive infinite debt rollover.
\end{proof}

\subsection{Proof of Proposition \ref{prop:bubble_necessity}}\label{subsec:proof_bubble_necessity}

Consider an auxiliary economy with one young agent, one old agent, one unit of land, aggregate endowments $(a_t,b_t)$ defined in \S\ref{subsec:bubble_necessity}, aggregate utility $\widetilde U_t$, and the same land rent $r_t$. To verify the equivalence with the original economy, define $Y_t\coloneqq N_tc_t^y$, $Z_{t+1}\coloneqq N_tc_{t+1}^o$, and $\widetilde x_t\coloneqq N_tx_t$. Multiplying the two budget constraints in \eqref{eq:budget} by $N_t$ gives
\begin{align*}
    &\text{Young:} &Y_t+P_t\widetilde x_t&=a_t,\\
    &\text{Old:} &Z_{t+1}&=b_{t+1}+(P_{t+1}+r_{t+1})\widetilde x_t.
\end{align*}
Land-market clearing becomes $\widetilde x_t=1$, and $\widetilde U_t(Y_t,Z_{t+1})=U_t(c_t^y,c_{t+1}^o)$. Multiplying the budget constraint of the initial old by $N_{-1}$ gives the initial old's aggregate consumption $b_0+P_0+r_0$. The transformation therefore preserves each generation's preferences and budget set, market clearing, the land-price sequence, and the initial ownership of land. Equilibria of the two economies correspond one to one.

Assumption \ref{asmp:preferences} implies the utility regularity condition of \citet{HiranoToda2025JPE}, while Assumptions \ref{asmp:long_run_growth} and \ref{asmp:forward_rates} are precisely their endowment-growth, relative-endowment, and forward-rate conditions. Theorem~1 of \citet{HiranoToda2025JPE} therefore guarantees that the auxiliary economy has an equilibrium. Under \eqref{eq:bubble_necessity}, their Theorem~2 implies that every equilibrium is bubbly and $P_t/a_t$ is bounded away from zero. The one-to-one correspondence transfers both conclusions to the original economy. The latter condition is exactly our definition of an asymptotically bubbly equilibrium, which completes the proof. \hfill \qedsymbol

\subsection{Proof of Theorem~\ref{thm:necessity}}\label{subsec:proof_necessity}

We use the following standard version of the Cass criterion. Consider a constant-population two-period OLG endowment economy in which the young consume $y_t$, the old consume $z_t$, and aggregate resources $E_t$ are uniformly bounded. An allocation is feasible if $y_t+z_t\le E_t$ for every $t$. Generation $t\ge0$ has utility $U_t(y_t,z_{t+1})$, while the initial old strictly prefer more date-$0$ consumption. Let $q_t>0$ denote the date-$0$ valuation of the date-$t$ good. The Cass criterion is
\begin{equation}
    \sum_{t=0}^\infty\frac{1}{q_t}=\infty. \label{eq:Cass1}
\end{equation}

\begin{prop}[Cass criterion]\label{prop:Cass2}
Let $\set{(c_t^y,c_t^o)}_{t=0}^\infty$ be an interior competitive equilibrium allocation of this economy. Suppose every $U_t$ is twice continuously differentiable, differentiably strictly increasing, and differentiably strictly quasi-concave. For each $t$, let $z=\phi_t(y)$ denote the indifference curve through $(c_t^y,c_{t+1}^o)$. Suppose there exist constants $\eta>0$ and $\mu>0$ such that every $y>0$ satisfying $\abs{y-c_t^y}\le\eta$ belongs to the domain of $\phi_t$ and
\begin{equation}
    \phi_t''(y)
    \ge-\frac{2\mu}{c_t^y}\phi_t'(c_t^y). \label{eq:mu}
\end{equation}
If \eqref{eq:Cass1} holds, then the equilibrium allocation is Pareto efficient.
\end{prop}

\citet{Cass1972} introduced the criterion \eqref{eq:Cass1}, and \citet{Benveniste1976}, \citet{BalaskoShell1980}, and \citet{OkunoZilcha1980}, among others, adapted it to other infinite-horizon settings. Appendix~\ref{sec:Cass} derives the curvature condition and proves Proposition \ref{prop:Cass2}.

\begin{proof}[Proof of Theorem~\ref{thm:necessity}]
Assumption \ref{asmp:welfare} implies Assumptions \ref{asmp:preferences} and \ref{asmp:forward_rates}. Proposition \ref{prop:bubble_necessity} therefore establishes that an equilibrium exists and that every equilibrium is asymptotically bubbly. Fix an equilibrium. We establish efficiency in two steps and then apply the debt--bubble equivalence.

\begin{step}\label{step:detrending}
Detrending and the Cass criterion.
\end{step}

Let $\Lambda_t\coloneqq N_t\lambda_t$ and $n_t\coloneqq N_{t+1}/N_t$. We transform the original economy into a constant-population economy by defining
\begin{subequations}
\begin{align}
    (\widehat e_t^y,\widehat e_t^o,d_t)
    &\coloneqq (a_t,b_t,r_t)/\Lambda_t, \\
    (\widehat c_t^y,\widehat c_t^o,p_t)
    &\coloneqq
    (N_tc_t^y,N_{t-1}c_t^o,P_t)/\Lambda_t, \label{eq:detrended_allocation}\\
    \widehat U_t(y,z)&\coloneqq U(y,n_tz). \label{eq:detrended_utility}
\end{align}
\end{subequations}
In particular, $\widehat e_t^y=e_t^y/\lambda_t$ and $\widehat e_t^o=(N_{t-1}/N_t)e_t^o/\lambda_t$. Dividing the aggregate budget constraints in the proof of Proposition \ref{prop:bubble_necessity} by $\Lambda_t$ and $\Lambda_{t+1}$ gives
\begin{subequations}
\begin{align}
    \widehat c_t^y+p_t&=\widehat e_t^y, \label{eq:detrended_budget_young}\\
    \widehat c_{t+1}^o&=\widehat e_{t+1}^o+p_{t+1}+d_{t+1}.
\end{align}
\end{subequations}
At date $t$, therefore,
\begin{equation}
    \widehat c_t^y+\widehat c_t^o
    =\widehat e_t^y+\widehat e_t^o+d_t. \label{eq:detrended_resource}
\end{equation}
Moreover,
\begin{equation*}
    \widehat U_t(\widehat c_t^y,\widehat c_{t+1}^o)
    =U(c_t^y/\lambda_t,c_{t+1}^o/\lambda_{t+1})
    =U_t(c_t^y,c_{t+1}^o),
\end{equation*}
so the transformation preserves preferences as well as feasibility.

Assumption \ref{asmp:welfare} implies $\Lambda_{t+1}/\Lambda_t\to G_NG_e$ and $(\widehat e_t^y,\widehat e_t^o)\to (e^y,e^o/G_N)$. Because $a_t/\Lambda_t=\widehat e_t^y\to e^y>0$, Assumption \ref{asmp:long_run_growth} further implies $G=G_NG_e$. Hence $\Lambda_t^{1/t}\to G$. Choose numbers $G_1$ and $G_2$ such that $G_r<G_1<G_2<G$. For all sufficiently large $t$, we have $r_t<G_1^t$ and $\Lambda_t>G_2^t$. Consequently, $0\le d_t=r_t/\Lambda_t<(G_1/G_2)^t\to 0$. Thus the resources in \eqref{eq:detrended_resource} are uniformly bounded.

The gross return in the detrended economy is $\widehat R_t=(p_{t+1}+d_{t+1})/p_t=R_t\Lambda_t/\Lambda_{t+1}$. Its date-$0$ valuations, normalized by $\widehat q_0=1$, are therefore $\widehat q_t=q_t\Lambda_t/\Lambda_0$. Proposition~\ref{prop:bubble_necessity} and $a_t/\Lambda_t\to e^y>0$ give
\begin{equation}
    \liminf_{t\to\infty}p_t
    =\liminf_{t\to\infty}
    \frac{P_t}{a_t}\frac{a_t}{\Lambda_t}>0. \label{eq:detrended_price_lower_bound}
\end{equation}
Let $B_0=P_0-V_0$ denote the date-$0$ bubble component. Since the equilibrium is bubbly, \eqref{eq:Pt_lim} gives $q_tP_t\to B_0>0$. The expression for $\widehat q_t$ implies the identity $1/\widehat q_t=\Lambda_0p_t/(q_tP_t)$. Together with $q_tP_t\to B_0$ and \eqref{eq:detrended_price_lower_bound}, this identity yields $\liminf_{t\to\infty}1/\widehat q_t>0$. In particular, $\sum_{t=0}^\infty1/\widehat q_t=\infty$, so the detrended equilibrium satisfies the Cass criterion.

\begin{step}\label{step:efficiency}
Uniform curvature and Pareto efficiency.
\end{step}

It remains to verify condition \eqref{eq:mu} in Proposition \ref{prop:Cass2}. The young-agent budget constraint in \eqref{eq:detrended_budget_young} gives $0<p_t<\widehat e_t^y$. Because $\widehat e_t^y\to e^y>0$ and \eqref{eq:detrended_price_lower_bound} holds, $p_t$ is bounded above and uniformly bounded away from zero. The old-agent budget constraint then implies $\widehat c_{t+1}^o=\widehat e_{t+1}^o+p_{t+1}+d_{t+1}\ge p_{t+1}$. Together with the resource bound in \eqref{eq:detrended_resource}, this inequality places $\set{\widehat c_{t+1}^o}$ in a compact interval $J_o\subset\R_{++}$.

The sequence $n_t$ converges to $G_N>0$ and hence lies in a compact interval $J_N\subset\R_{++}$. Moreover, $\widehat R_t=(p_{t+1}+d_{t+1})/p_t$ is uniformly bounded above because its numerator is bounded and its denominator is uniformly bounded away from zero. The household first-order condition in the detrended economy gives
\begin{equation*}
    \widehat R_t
    =(\widehat U_{t,1}/\widehat U_{t,2})
    (\widehat c_t^y,\widehat c_{t+1}^o)
    =(1/n_t)\cR
    (\widehat c_t^y,n_t\widehat c_{t+1}^o).
\end{equation*}
Since $n_t\widehat c_{t+1}^o$ lies in a compact interval in $\R_{++}$, the uniform Inada condition in Assumption \ref{asmp:welfare} and the upper bound on $\widehat R_t$ imply that $\widehat c_t^y$ is uniformly bounded away from zero. The resource bound supplies a uniform upper bound. Hence there exists a compact set $K_c\subset\R_{++}^2$ such that $(\widehat c_t^y,\widehat c_{t+1}^o)\in K_c$ for every $t\ge0$. Define $\cK\coloneqq J_N\times K_c$, which is a compact subset of $\R_{++}^3$.

We now translate this compactness result into the curvature bound in Proposition \ref{prop:Cass2}. For $(n,c)\in\cK$, let $z=\phi_{n,c}(y)$ denote the indifference curve of $(y,z)\mapsto U(y,nz)$ through $c$. Differentiable strict quasi-concavity implies $\phi_{n,c}''(y)>0$ wherever this curve lies in the interior, while strict monotonicity implies $-\phi_{n,c}'(c^y)>0$. The implicit function theorem and compactness yield a common $\eta>0$ such that, for every $(n,c)\in\cK$, the graph is well defined when $y>0$ and $\abs{y-c^y}\le\eta$ and remains in a common compact subset of $\R_{++}^2$. On the resulting compact set, the function $(n,c,y)\mapsto -c^y\phi_{n,c}''(y)/\phi_{n,c}'(c^y)$ is continuous and strictly positive. It therefore has a strictly positive minimum. Choosing $2\mu$ no larger than that minimum verifies \eqref{eq:mu}. Proposition \ref{prop:Cass2} therefore implies that the detrended equilibrium is Pareto efficient. Since the transformation preserves feasibility and preferences, the original equilibrium is also Pareto efficient. Finally, the debt-rollover claim follows from $B_0>0$ and Lemma \ref{lem:debt_bubble}.
\end{proof}

\subsection{Proof of Proposition \ref{prop:necessity}}\label{subsec:proof_CES}

At $X=1$, differentiating the CES production function \eqref{eq:CES} gives
\begin{align*}
    w_t/A_t
    &=\alpha(\alpha H_t^{1-1/\sigma}+1-\alpha)^{\frac{1}{\sigma-1}}H_t^{-1/\sigma},\\
    r_t/A_t
    &=(1-\alpha)(\alpha H_t^{1-1/\sigma}+1-\alpha)^{\frac{1}{\sigma-1}}.
\end{align*}
Because $N_{t-1}/N_t\to1/G_N$, \eqref{eq:scaled_labor_endowments} implies $H_t/(N_t\ell_t)\to\bar h$. Condition $G_NG_H>1$ then implies $H_t\to\infty$. Factoring $H_t^{1-1/\sigma}$ out of the differentiated factor-price formulas establishes \eqref{eq:CES_limits}. By \eqref{eq:scaled_labor_endowments}, the induced endowments satisfy
\begin{equation*}
    \left(\frac{e_t^y}{A_t\ell_t},\frac{e_t^o}{A_t\ell_t}\right)
    =\frac{w_t}{A_t}
    \left(\frac{h_t^y}{\ell_t},\frac{h_t^o}{\ell_t}\right)
    \to(e^y,e^o)\coloneqq(\bar wh^y,\bar wh^o).
\end{equation*}
Since $a_t=N_te_t^y$ and $b_t=N_{t-1}e_t^o$, we obtain $a_{t+1}/a_t\to G_AG_NG_H$ and $b_t/a_t\to e^o/(G_Ne^y)$. Thus Assumption \ref{asmp:long_run_growth} holds with $G=G_AG_NG_H$ and $\omega=e^o/(G_Ne^y)$. Under the calendar-date normalization with $\lambda_t=A_t\ell_t$, the normalized forward-rate function is
\begin{equation*}
    f_t(y,z)
    =\frac{A_{t+1}\ell_{t+1}}{A_t\ell_t}
    \cR\left(\frac{e_t^y}{A_t\ell_t}y,
    \frac{e_t^y}{A_{t+1}\ell_{t+1}}z\right).
\end{equation*}
It converges uniformly on compact subsets of $\R_{++}\times\R_+$ to
\begin{equation*}
    f(y,z)=G_AG_H\cR\left(e^yy,\frac{e^y}{G_AG_H}z\right).
\end{equation*}
Hence Assumption \ref{asmp:forward_rates} holds, and \eqref{eq:Ra} gives $R^{\mathrm a}=f(1,G\omega)=G_AG_H\cR(e^y,e^o)$. Assumption \ref{asmp:welfare} holds with $\lambda_t=A_t\ell_t$ and $G_e=G_AG_H$. Finally, \eqref{eq:CES_growth} and \eqref{eq:CES_rent_growth} imply $G_r=G_A(G_NG_H)^{1/\sigma}<G_AG_NG_H=G$. Condition \eqref{eq:CES_bubble_necessity} is equivalent to $R^{\mathrm a}<G_r$. Proposition \ref{prop:bubble_necessity} and Theorem~\ref{thm:necessity} therefore give all the stated conclusions. \hfill \qedsymbol

\printbibliography

\newpage

\begin{refsection}

\begin{center}
    {\Large \textbf{Online Appendix (not for publication)}}
\end{center}

\section{Supplementary theoretical results}\label{sec:supplementary_theory}

This appendix collects supplementary theoretical results. Appendix~\ref{sec:bubble} reviews fundamental results on asset price bubbles, Appendix~\ref{sec:Cass} proves the Cass criterion used in the proof of Theorem~\ref{thm:necessity}, and Appendix~\ref{sec:equilibrium_constructions} constructs the illustrative equilibria in Proposition~\ref{prop:constructed_equilibria}.

\subsection{Fundamental results on asset price bubbles}\label{sec:bubble}

This appendix reviews fundamental results on asset price bubbles \citep{SantosWoodford1997,HiranoToda2024JME,HiranoToda2025JPE}.

Consider a deterministic infinite-horizon economy with time indexed by $t=0,1,\dotsc$. Let $q_t>0$ be the date-$0$ valuation of the date-$t$ good, normalized by $q_0=1$. Suppose an asset pays a nonnegative dividend $D_t$ and trades at price $P_t>0$ at date $t$. The absence of arbitrage implies $q_tP_t=q_{t+1}(P_{t+1}+D_{t+1})$.
Iterating over time, for any $T>t$ we obtain
\begin{equation}
    q_tP_t=\sum_{s=t+1}^T q_sD_s+q_TP_T. \label{eq:Pt_iter}
\end{equation}
Because $q_t>0$, $P_t\ge0$, and $D_t\ge0$ for all $t$, the partial sum $\set{\sum_{s=t+1}^Tq_sD_s}$ is increasing in $T$ and bounded above by $q_tP_t$. It therefore converges. Letting $T\to\infty$ in \eqref{eq:Pt_iter} and dividing by $q_t>0$, we obtain
\begin{equation}
    P_t=V_t+\frac{1}{q_t}\lim_{T\to\infty}q_TP_T,
    \qquad
    V_t\coloneqq\frac{1}{q_t}\sum_{s=t+1}^\infty q_sD_s. \label{eq:Pt_lim}
\end{equation}
Thus, $P_t\ge V_t$ for all $t$. We say that the asset price contains a \emph{bubble} at date $t$ if $P_t>V_t$. The asset price is fundamental, $P_t=V_t$ for all $t$, if and only if
\begin{equation}
    \lim_{T\to\infty}q_TP_T=0. \label{eq:nobubble}
\end{equation}
We refer to \eqref{eq:nobubble} as the \emph{no-bubble condition}.

The following lemma provides a useful characterization of bubbles.

\begin{lem}[Bubble Characterization, \citealp{Montrucchio2004}, Proposition 7]\label{lem:bubble}
In any equilibrium, the asset price exhibits a bubble if and only if
\begin{equation*}
    \sum_{t=0}^\infty\frac{D_t}{P_t}<\infty.
\end{equation*}
\end{lem}

Lemma \ref{lem:bubble} says that a bubble emerges if and only if the price-dividend ratio $P_t/D_t$ grows sufficiently fast. See also \citet[Lemma 1]{HiranoToda2025JPE}.

\subsection{Cass criterion}\label{sec:Cass}

This appendix proves Proposition \ref{prop:Cass2} and records its relation to the literature. We first make the relevant curvature condition explicit. Let $U:\R_{++}^L\to\R$ be twice continuously differentiable. We say that $U$ is \emph{differentiably strictly increasing} if $\nabla U(c)\gg0$ for all $c\in\R_{++}^L$. We say that $U$ is \emph{differentiably strictly quasi-concave} if, for every $c\in\R_{++}^L$ and every nonzero $v\in\R^L$,
\begin{equation*}
    \seq{\nabla U(c),v}=0
    \quad\Longrightarrow\quad
    \seq{v,\nabla^2U(c)v}<0,
\end{equation*}
where $\seq{\cdot,\cdot}$ denotes the inner product in $\R^L$. See \citet[p.~14, Definition 29]{VillanacciCarosiBenevieriBattinelli2002} and \citet[p.~155, Proposition 11.11(ii)]{TodaEME}.

In the two-good case, these conditions make every indifference curve through an interior bundle the graph of a strictly decreasing, strictly convex function. To verify this claim, write the curve as $z=\phi(y)$. Implicit differentiation gives $\phi'(y)=-U_1/U_2<0$. Differentiating once more and setting $v=(1,-U_1/U_2)$ gives
\begin{equation}
    \phi''(y)
    =-\frac{U_{11}-2U_{12}(U_1/U_2)+U_{22}(U_1/U_2)^2}{U_2}
    =-\frac{\seq{v,\nabla^2Uv}}{U_2}>0. \label{eq:curvature}
\end{equation}
Thus, the local curvature hypothesis \eqref{eq:mu} imposes a uniform lower bound on the curvature of the relevant indifference curves near the equilibrium allocation. It does not impose a global restriction on indifference curves far from that allocation.

\begin{proof}[Proof of Proposition \ref{prop:Cass2}]
The argument follows Theorem 3A of \citet{OkunoZilcha1980}, which builds on \citet{Benveniste1976}. We give the details because the uniformity in \eqref{eq:mu} matters for our application.

We prove the contrapositive. Suppose the equilibrium allocation is Pareto inefficient, and let $\set{(y_t,z_t)}_{t=0}^\infty$ be a feasible Pareto improvement. We may take this improvement to be uniformly close to the equilibrium allocation. To see this, for $\lambda\in(0,1)$ define
\begin{equation*}
    (y_t^\lambda,z_t^\lambda)
    =(1-\lambda)(c_t^y,c_t^o)+\lambda(y_t,z_t).
\end{equation*}
Convexity preserves feasibility. Strict quasi-concavity for generations and monotonicity for the initial old imply that this allocation weakly improves every agent and strictly improves at least one. Because both allocations are feasible and $E_t\le\bar E\coloneqq\sup_tE_t$, choosing $\lambda\bar E<\eta$ ensures
\begin{equation*}
    \abs{y_t^\lambda-c_t^y}<\eta,
    \qquad
    \abs{z_t^\lambda-c_t^o}<\eta
\end{equation*}
for every $t$. We henceforth relabel this nearby Pareto improvement as $(y_t,z_t)$.

Because generation $t$ weakly prefers $(y_t,z_{t+1})$ to $(c_t^y,c_{t+1}^o)$ and $U_t$ is strictly increasing in old-age consumption, $z_{t+1}\ge\phi_t(y_t)$, with a strict inequality if generation $t$ is strictly better off. Taylor's theorem gives, for some $\xi_t$ between $c_t^y$ and $y_t$,
\begin{equation*}
    \phi_t(y_t)
    =c_{t+1}^o+\phi_t'(c_t^y)(y_t-c_t^y)
    +\frac{1}{2}\phi_t''(\xi_t)(y_t-c_t^y)^2.
\end{equation*}
The choice of the nearby improvement ensures $\abs{\xi_t-c_t^y}<\eta$. Hence \eqref{eq:mu} implies
\begin{equation}
    z_{t+1}-c_{t+1}^o
    \ge \phi_t'(c_t^y)(y_t-c_t^y)
    -\frac{\mu}{c_t^y}\phi_t'(c_t^y)(y_t-c_t^y)^2. \label{eq:Taylor_bound}
\end{equation}
Individual optimality at the interior equilibrium gives $\phi_t'(c_t^y)=-q_t/q_{t+1}$. Multiplying \eqref{eq:Taylor_bound} by $q_{t+1}$ therefore yields
\begin{equation}
    q_t(y_t-c_t^y)+q_{t+1}(z_{t+1}-c_{t+1}^o)
    \ge \frac{\mu}{q_tc_t^y}\left[q_t(y_t-c_t^y)\right]^2. \label{eq:OZ1}
\end{equation}
The inequality is strict when generation $t$ is strictly better off. Equation \eqref{eq:OZ1} is the counterpart of Equation (1) in \citet{OkunoZilcha1980}.\footnote{Their equation contains a typographical error: $p_t^*(t)(c_t-c_t^*)$ on the right-hand side should be $p_t^*(t)(c_t^*-c_t)$.}

Feasibility implies $0<c_t^y\le E_t\le\bar E$. Thus \eqref{eq:OZ1} implies
\begin{equation}
    q_t(y_t-c_t^y)+q_{t+1}(z_{t+1}-c_{t+1}^o)
    \ge \frac{\kappa}{q_t}\left[q_t(y_t-c_t^y)\right]^2, \label{eq:OZ2}
\end{equation}
where $\kappa\coloneqq\mu/\bar E>0$. Define the value changes
\begin{equation*}
    \delta_t\coloneqq q_t(y_t-c_t^y),
    \qquad
    \epsilon_t\coloneqq q_t(z_t-c_t^o).
\end{equation*}
Date-$t$ feasibility gives
\begin{equation}
    \delta_t+\epsilon_t\le0. \label{eq:de1}
\end{equation}
The initial old are weakly better off, so $\epsilon_0\ge0$. Moreover, \eqref{eq:OZ2} gives
\begin{equation}
    \delta_t+\epsilon_{t+1}
    \ge\frac{\kappa}{q_t}\delta_t^2\ge0. \label{eq:de2}
\end{equation}
If the initial old are strictly better off, then $\epsilon_0>0$; if generation $t$ is strictly better off, then the first inequality in \eqref{eq:de2} is strict for that $t$.

Combining \eqref{eq:de1} and \eqref{eq:de2} gives
\begin{align*}
    \delta_{t+1}-\delta_t
    &=(\delta_{t+1}+\epsilon_{t+1})-(\delta_t+\epsilon_{t+1})\le0,\\
    \epsilon_{t+1}-\epsilon_t
    &=(\delta_t+\epsilon_{t+1})-(\delta_t+\epsilon_t)\ge0.
\end{align*}
Consequently,
\begin{equation*}
    \delta_t\le\dotsb\le\delta_0\le-\epsilon_0\le0
    \le\epsilon_0\le\dotsb\le\epsilon_t.
\end{equation*}
At least one $\epsilon_t$ must be strictly positive. Otherwise, \eqref{eq:de1}--\eqref{eq:de2} would imply $\delta_t=0$ for every $t$, contradicting the fact that some agent is strictly better off. Monotonicity therefore implies that there exists $t_0$ such that $-\delta_t\ge\epsilon_t>0$ for $t\ge t_0$. Using these inequalities in \eqref{eq:de2}, for $t\ge t_0$ we obtain
\begin{equation}
    \epsilon_{t+1}
    \ge-\delta_t+\frac{\kappa}{q_t}\delta_t^2
    \ge\epsilon_t+\frac{\kappa}{q_t}\epsilon_t^2. \label{eq:epsilon_recursion}
\end{equation}

Set $h_t\coloneqq\epsilon_t/q_t=z_t-c_t^o>0$ for $t\ge t_0$. Equation \eqref{eq:epsilon_recursion} becomes
\begin{equation*}
    q_{t+1}h_{t+1}\ge q_th_t(1+\kappa h_t).
\end{equation*}
Taking reciprocals gives
\begin{equation*}
    \frac{1}{q_{t+1}h_{t+1}}
    \le\frac{1}{q_th_t}
    -\frac{\kappa}{1+\kappa h_t}\frac{1}{q_t}.
\end{equation*}
Because $0<z_t\le E_t\le\bar E$ and $c_t^o>0$, we have $0<h_t\le\bar E$. Hence
\begin{equation*}
    \frac{\kappa}{1+\kappa\bar E}
    \sum_{t=t_0}^T\frac{1}{q_t}
    \le
    \frac{1}{q_{t_0}h_{t_0}}
    -\frac{1}{q_{T+1}h_{T+1}}
    \le\frac{1}{q_{t_0}h_{t_0}}.
\end{equation*}
Letting $T\to\infty$ gives $\sum_{t=t_0}^\infty1/q_t<\infty$. Thus Pareto inefficiency violates \eqref{eq:Cass1}, which proves the contrapositive.
\end{proof}

\begin{rem}\label{rem:elasticity}
Condition \eqref{eq:mu} is a uniform local-curvature condition. It follows, for example, if preferences belong to a compact twice continuously differentiable family, the relevant equilibrium bundles and nearby indifference curves lie in a common compact subset of $\R_{++}^2$, and every utility function is differentiably strictly quasi-concave. Compactness and \eqref{eq:curvature} then provide the two uniform constants $\eta$ and $\mu$. The proof of Theorem~\ref{thm:necessity} verifies these properties for the calendar-date-normalized preferences in \eqref{eq:normalized_preferences}. \citet[p.~805]{OkunoZilcha1980} provide a counterexample when the required uniform curvature fails.
\end{rem}

\begin{rem}
Although Proposition \ref{prop:Cass2} considers a two-period OLG model with one commodity in each period, that restriction is not essential. \citet{Benveniste1986} derives an analog of \eqref{eq:OZ1} in a general setting.
\end{rem}

\subsection{Illustrative equilibrium constructions}\label{sec:equilibrium_constructions}

We now construct two equilibria that clarify the scope of Theorem~\ref{thm:necessity}. By Lemma \ref{lem:production_endowment}, these reduced-form endowment economies have linear-production representations through \eqref{eq:linear_representation}; they are not additional CES production economies. The examples do not provide additional general results. Instead, they isolate the role of the land-price growth rate by reversing the usual order of equilibrium determination: we first specify a price sequence and then recover the young agents' endowments from their Euler equations.

We use the endowment-economy representation in Lemma \ref{lem:production_endowment}. Let $N_t=G^t$ with $G>1$, impose the CRRA preferences in \eqref{eq:CRRA_preferences}, and fix $\sigma>1$ and $r>0$. To match the rent growth generated by the constant-productivity specialization of the CES economy in \S\ref{subsec:CES_rollover}, let land pay rent
\begin{equation}
    r_t=rG^{t/\sigma}. \label{eq:constructed_rent}
\end{equation}
Fix $e^o\ge0$ and set $e_t^o=e^o$. For any positive candidate land-price sequence $\set{P_t}$, define
\begin{subequations}\label{eq:reverse_engineering}
\begin{align}
    R_t&\coloneqq\frac{P_{t+1}+r_{t+1}}{P_t}, \label{eq:constructed_return}\\
    z_{t+1}&\coloneqq e^o+\frac{P_{t+1}+r_{t+1}}{G^t}, \label{eq:constructed_old_consumption}\\
    e_t^y&\coloneqq\frac{P_t}{G^t}
    +z_{t+1}(\beta R_t)^{-1/\gamma}. \label{eq:constructed_young_endowment}
\end{align}
\end{subequations}
The converse in Lemma \ref{lem:production_endowment} realizes this endowment economy as a production economy through \eqref{eq:linear_representation}: set $F_t(H,X)=H+r_tX$ and $(h_t^y,h_t^o)=(e_t^y,e^o)$. Thus, land remains a productive, non-reproducible asset in both examples.

\begin{prop}[Bubbly and fundamental constructions]\label{prop:constructed_equilibria}
Consider the economy defined by \eqref{eq:constructed_rent} and \eqref{eq:reverse_engineering}.
\begin{enumerate}
    \item If $P_t=pG^t$ for some $p>0$, the constructed equilibrium is asymptotically bubbly and Pareto efficient. Let $B_0$ denote its date-$0$ bubble component. For every initial issuance $D_0\in(0,B_0]$, there exists a corresponding equilibrium with positive infinite debt rollover. If, in addition,
    \begin{equation}
        \frac{1}{\beta}
        \left(
        \frac{e^o}{p+(e^o+pG)(\beta G)^{-1/\gamma}}
        \right)^\gamma
        <G^{1/\sigma}, \label{eq:constructed_bubble_necessity}
    \end{equation}
    then every equilibrium of the constructed economy is asymptotically bubbly and Pareto efficient.

    \item Suppose $e^o>0$. If $P_t=pG^{t/\sigma}$ for some $p>0$, the constructed equilibrium is fundamental. Its constant gross interest rate is
    \begin{equation}
        R=G^{1/\sigma}(1+r/p), \label{eq:constructed_fundamental_return}
    \end{equation}
    and the equilibrium is Pareto efficient if and only if
    \begin{equation}
        R\ge G
        \quad\Longleftrightarrow\quad
        p\le\frac{r}{G^{1-1/\sigma}-1}. \label{eq:constructed_efficiency}
    \end{equation}
\end{enumerate}
\end{prop}

\begin{proof}
We divide the proof into three steps.

\begin{step}\label{step:verify_construction}
Equilibrium verification.
\end{step}

For either candidate price sequence, \eqref{eq:constructed_young_endowment} implies $e_t^y>P_t/G^t>0$. If each young agent purchases $x_t=1/G^t$ units of land, the budget constraints and \eqref{eq:constructed_old_consumption}--\eqref{eq:constructed_young_endowment} give $(c_t^y,c_{t+1}^o)=(z_{t+1}(\beta R_t)^{-1/\gamma},z_{t+1})\gg0$. Consequently, the marginal rate of substitution satisfies
\begin{equation*}
    \frac{U_1}{U_2}(c_t^y,c_{t+1}^o)
    =\frac{1}{\beta}\left(\frac{c_{t+1}^o}{c_t^y}\right)^\gamma
    =R_t.
\end{equation*}
Equation \eqref{eq:constructed_return} therefore implies the household Euler equation. Strict concavity of CRRA utility makes this first-order condition sufficient for optimality. The choices $x_t=1/G^t$ clear the land market, and summing the household budget constraints shows that the goods market also clears. Hence \eqref{eq:reverse_engineering} supports an equilibrium for either price sequence.

\begin{step}\label{step:bubbly_construction}
The bubbly construction.
\end{step}

By Step~\ref{step:verify_construction}, it remains to characterize bubbles and efficiency. Let $P_t=pG^t$ and define $\chi\coloneqq G^{1/\sigma-1}\in(0,1)$. Equations \eqref{eq:constructed_return} and \eqref{eq:constructed_old_consumption} imply
\begin{subequations}\label{eq:constructed_bubbly_Rz}
\begin{align}
    R_t&=G+\frac{r}{p}G^{1/\sigma}\chi^t\to G, \label{eq:constructed_bubbly_return}\\
    z_{t+1}&=e^o+pG+rG^{1/\sigma}\chi^t\to e^o+pG. \label{eq:constructed_bubbly_old_consumption}
\end{align}
\end{subequations}
It follows from \eqref{eq:constructed_young_endowment} and \eqref{eq:constructed_bubbly_Rz} that
\begin{equation}
    e_t^y\to \bar e^y
    \coloneqq p+(e^o+pG)(\beta G)^{-1/\gamma}>0. \label{eq:constructed_bubbly_endowment_limit}
\end{equation}
Thus, the economy satisfies Assumptions \ref{asmp:long_run_growth} and \ref{asmp:forward_rates}, as well as Assumption \ref{asmp:welfare} with $\lambda_t=1$, $G_N=G$, and $G_e=1$. Moreover, \eqref{eq:constructed_rent} and $P_t=pG^t$ give $r_t/P_t=(r/p)\chi^t$. Lemma \ref{lem:bubble} therefore implies that the equilibrium is bubbly, while $P_t/G^t=p$ makes it asymptotically bubbly.

We next verify efficiency. Let $q_t$ denote the date-$0$ valuation sequence and detrend the economy as in Step~\ref{step:detrending}. The detrended valuations satisfy $\widehat q_t=q_tG^t$, and \eqref{eq:constructed_bubbly_return} gives $1/\widehat q_t=\prod_{s=0}^{t-1}(R_s/G)\ge1$. Hence the Cass criterion \eqref{eq:Cass1} holds. Equations \eqref{eq:constructed_bubbly_return}--\eqref{eq:constructed_bubbly_endowment_limit} place equilibrium consumption and endowments, after detrending, in a common compact subset of the interior. CRRA utility then satisfies the uniform-curvature condition \eqref{eq:mu}. The argument in Step~\ref{step:efficiency} and Proposition \ref{prop:Cass2} imply that the equilibrium is Pareto efficient. Because the equilibrium is bubbly, Lemma \ref{lem:debt_bubble} allows the government to replace a positive portion of its bubble with debt and roll that debt over forever.

Finally, \eqref{eq:Ra} and \eqref{eq:constructed_bubbly_endowment_limit} imply that the autarky interest rate equals the left-hand side of \eqref{eq:constructed_bubble_necessity}. The rent growth rate is $G_r=G^{1/\sigma}$. Under \eqref{eq:constructed_bubble_necessity}, Proposition \ref{prop:bubble_necessity} implies that every equilibrium is asymptotically bubbly, and Theorem~\ref{thm:necessity} implies that every equilibrium is Pareto efficient.

\begin{step}\label{step:fundamental_construction}
The fundamental construction.
\end{step}

Suppose $e^o>0$ and let $P_t=pG^{t/\sigma}$. Equations \eqref{eq:constructed_return}--\eqref{eq:constructed_young_endowment} imply
\begin{subequations}
\begin{align}
    R_t&=G^{1/\sigma}(1+r/p)=R, \label{eq:constructed_constant_return}\\
    z_{t+1}&=e^o+(p+r)G^{1/\sigma}\chi^t\to e^o, \label{eq:constructed_fundamental_old_consumption}\\
    e_t^y&\to \bar e^y\coloneqq e^o(\beta R)^{-1/\gamma}>0. \label{eq:constructed_fundamental_endowment_limit}
\end{align}
\end{subequations}
Thus, this economy also satisfies Assumptions \ref{asmp:long_run_growth} and \ref{asmp:forward_rates}, as well as Assumption \ref{asmp:welfare} with $\lambda_t=1$, $G_N=G$, and $G_e=1$. By \eqref{eq:Ra}, the autarky interest rate satisfies $R^{\mathrm a}=(1/\beta)(e^o/\bar e^y)^\gamma=R>G^{1/\sigma}=G_r$, so the sufficient condition in Proposition \ref{prop:bubble_necessity} does not hold. Since $q_t=R^{-t}$, the fundamental value of land is
\begin{equation*}
	V_t=\sum_{k=1}^\infty R^{-k}rG^{(t+k)/\sigma}=rG^{t/\sigma}\frac{G^{1/\sigma}}{R-G^{1/\sigma}}
    =pG^{t/\sigma}=P_t,
\end{equation*}
where the penultimate equality uses \eqref{eq:constructed_fundamental_return}. The equilibrium is therefore fundamental.

To characterize efficiency, detrend as in Step~\ref{step:detrending}. Because $\widehat q_t=q_tG^t=(G/R)^t$, the Cass criterion becomes
\begin{equation*}
    \sum_{t=0}^\infty\frac{1}{\widehat q_t}
    =\sum_{t=0}^\infty(R/G)^t=\infty
\end{equation*}
if and only if $R\ge G$. When $R\ge G$, \eqref{eq:constructed_fundamental_old_consumption} and \eqref{eq:constructed_fundamental_endowment_limit} again give the compactness and uniform curvature used in Step~\ref{step:efficiency}. Proposition \ref{prop:Cass2} therefore implies Pareto efficiency.

Conversely, suppose $R<G$. At every date $t$ after some date $T\ge1$, tax each young agent by $\epsilon>0$ and transfer $G\epsilon$ to each old agent. This transfer balances at every date because $N_t/N_{t-1}=G$. For generation $t\ge T$, define
\begin{equation*}
    f_t(\epsilon)
    \coloneqq U(c_t^y-\epsilon,c_{t+1}^o+G\epsilon).
\end{equation*}
The Euler equation and \eqref{eq:constructed_constant_return} give
\begin{equation*}
    f_t'(0)
    =-U_1(c_t^y,c_{t+1}^o)+GU_2(c_t^y,c_{t+1}^o)
    =(G-R)U_2(c_t^y,c_{t+1}^o)>0.
\end{equation*}
By \eqref{eq:constructed_fundamental_old_consumption} and \eqref{eq:constructed_fundamental_endowment_limit}, the consumption bundles converge to a point in $\R_{++}^2$. We may therefore choose $T$ sufficiently large and a common $\epsilon>0$ sufficiently small that $f_t(\epsilon)>f_t(0)$ for every $t\ge T$. Generation $T-1$ receives the old-age transfer without having paid the young-age tax, generations born before $T-1$ remain unaffected, and every generation born from $T$ onward strictly benefits. The transfer scheme is consequently a Pareto improvement, so the equilibrium is inefficient when $R<G$. This proves the first equivalence in \eqref{eq:constructed_efficiency}; substituting \eqref{eq:constructed_fundamental_return} gives the second. Together, Steps~\ref{step:bubbly_construction} and \ref{step:fundamental_construction} prove the proposition.
\end{proof}

\begin{rem}
These constructions also clarify the relation to \citet{Rhee1991}. Rhee imposes his Assumption~A on the behavior of the marginal product of land along an equilibrium path. Under that high-level restriction, his Proposition~2 permits a continuum of asymptotically bubbleless equilibria alongside a bubbly equilibrium when the Diamond steady state is inefficient, and his Proposition~3 classifies their efficiency. Rhee states that the proofs are equivalent to those of Propositions~1 and~2 of \citet{Tirole1985}. However, \citet{PhamToda2026ECMA} show that Tirole's Proposition~1(c)---the claimed bubble-necessity result for a dividend-paying asset---is incorrect as stated and restore it only under additional restrictions. Thus, Rhee's appeal to Tirole does not establish bubble necessity. In Proposition \ref{prop:constructed_equilibria}, bubble necessity in the first construction follows instead from Proposition \ref{prop:bubble_necessity} under the explicit condition \eqref{eq:constructed_bubble_necessity}; the second construction, for which $R^{\mathrm a}>G_r$, shows how a fundamental equilibrium can be efficient or inefficient.
\end{rem}

\section{Iowa farmland data}\label{sec:farmland_data}

This appendix documents the Iowa farmland data and reports additional empirical results. \S\ref{subsec:data_land}--\ref{subsec:data_statewide} explain the construction of land values and cash rents, the use of districts rather than counties, the imputation of missing rents, and the statewide measures. \S\ref{subsec:district_results} reports additional results on district rents and land prices, while \S\ref{subsec:gordon_details}--\ref{subsec:gordon_long} report additional details and results for the Gordon model of \S\ref{subsec:gordon_estimates}. Table~\ref{tab:data_sources} lists the data sources.

\begin{table}[!htb]
\centering
\caption{Data sources.}\label{tab:data_sources}
\small
\begin{tabular}{>{\raggedright\arraybackslash}p{3.6cm}>{\raggedright\arraybackslash}p{5.6cm}>{\raggedright\arraybackslash}p{2.2cm}>{\raggedright\arraybackslash}p{1.4cm}}
\toprule
Series & Source & Coverage & Use \\
\midrule
Farmland values per acre, state, district, and county & ISU, \href{https://www.extension.iastate.edu/agdm/wholefarm/xls/c2-70landvalues.xlsx}{\emph{Historical Iowa Farmland Value Survey}} & State 1941--2025; district and county 1950--2025 & $P_t$, $P_{it}$, $P^c_{jt}$ \\
Cash rents per acre of corn and soybean land, state and county & ISU, \href{https://www.extension.iastate.edu/agdm/wholefarm/html/c2-11.html}{\emph{Historical Cash Rental Rates for Iowa Survey by County}} (AgDM File C2-11) & 1978--2025 & $r_t$, $r^c_{jt}$ \\
Cash rents per acre of corn and soybean land, district & ISU, \href{https://www.extension.iastate.edu/agdm/wholefarm/html/c2-10.html}{\emph{Cash Rental Rates for Iowa Survey}} (AgDM File C2-10), annual reports & 2010--2026 & Validation \\
Harvested cropland by county & USDA NASS, \href{https://www.nass.usda.gov/Publications/AgCensus/2022/Full_Report/Volume_1,_Chapter_2_County_Level/Iowa/}{2022 Census of Agriculture, Iowa county-level data} & 2022 & $w_j$ \\
Corn and soybean shares of harvested cropland by county & USDA NASS, \href{https://www.nass.usda.gov/Publications/AgCensus/2022/Online_Resources/Ag_Census_Web_Maps/Data_download/index.php}{2022 Census of Agriculture county-map data} & 2022 & $w_j$ \\
Soil regions & \citet[Figure~3]{HarrisLordWeirich1980} & Fixed & $s(j)$ \\
County boundaries & US Census Bureau, \href{https://www.census.gov/geographies/mapping-files/time-series/geo/carto-boundary-file.html}{cartographic boundary file} (1:20,000,000) & 2010 & Adjacency \\
Consumer Price Index for All Urban Consumers, all items, not seasonally adjusted & FRED, series \href{https://fred.stlouisfed.org/series/CPIAUCNS}{CPIAUCNS} & Monthly, 1913--2025 & Deflator \\
Gross domestic product, nominal & FRED, series \href{https://fred.stlouisfed.org/series/GDPA}{GDPA} & Annual, 1929--2025 & $Y_t$ \\
10-year Treasury constant maturity rate & FRED, series \href{https://fred.stlouisfed.org/series/GS10}{GS10} & Monthly, 1953--2025 & Real interest rate \\
10-year expected inflation & Federal Reserve Bank of Cleveland via FRED, series \href{https://fred.stlouisfed.org/series/EXPINF10YR}{EXPINF10YR} & Monthly, 1982--2025 & Real interest rate \\
\bottomrule
\end{tabular}
\caption*{\footnotesize \emph{Notes:} ISU: Iowa State University. USDA NASS: US Department of Agriculture, National Agricultural Statistics Service. FRED: Federal Reserve Economic Data, Federal Reserve Bank of St.\ Louis. We approximate the corn-and-soybean acreage of each county by its harvested cropland times the sum of its corn and soybean shares of harvested cropland; the sum over counties is within 0.001 percent of the official state total. We deflate by the December value of the CPI, normalized to one in 2025. All sources were accessed in October 2026.}
\end{table}

\subsection{Land values}\label{subsec:data_land}

\paragraph{State and district values}
\citet{HarrisLordWeirich1980} document the methodology of the Iowa Land Value Survey, which ISU still uses.\footnote{See the history of the survey at \url{https://farmland.card.iastate.edu/history}.} ISU publishes state, district, and county values from 1950. Because ISU changed how it constructed and reported the values in the years before 1980, \citet[pp.~1, 18]{HarrisLordWeirich1980} re-estimated the values for 1950--1979 with the methodology described below, so the published series are consistent before and after 1980. Each year, ISU asks real estate brokers and other agricultural professionals to estimate the value of high-, medium-, and low-grade farmland in their counties. The survey receives about 500 responses per year, or about five per county. For each county, ISU averages the responses for each grade and weights the three grade averages by the estimated share of each grade in the county \citep[pp.~2--4]{HarrisLordWeirich1980}. We call the result the \emph{preliminary county value}. The published statewide value $P_t$ is the simple average of the 99 preliminary county values, and the published value $P_{it}$ of district $i$ is the simple average of the preliminary values of the $n_i$ counties in the district \citep[pp.~5--8]{HarrisLordWeirich1980}. Because these averages pool many responses (about 500 for the state and about 50 for a district), they are reliable survey measures. Both are direct survey averages, and $P_t=\sum_in_iP_{it}/99$. In our data, $\sum_in_iP_{it}/99$ differs from $P_t$ by 0.046 percent on average over 1950--2025 and by at most 0.54 percent (in 2011), which confirms that the two series average the same preliminary county values.

\paragraph{County values}
In contrast, the preliminary value of a single county rests on only about five responses, so it is too noisy to publish \citep[p.~8]{HarrisLordWeirich1980}. ISU therefore combines the survey with the Census of Agriculture and constructs the published county values in three steps \citep[pp.~8--17]{HarrisLordWeirich1980}. First, it groups the counties into ten soil regions with homogeneous corn suitability ratings, which summarize how slope, erosion, drainage, and other soil characteristics affect crop yields \citep{Walker1976}. ISU uses soil regions rather than districts because counties with similar soils tend to have similar land values \citep[pp.~11--13]{HarrisLordWeirich1980}. Let $s(j)$ denote the soil region that contains county $j$ and $\bar P_{st}$ the survey value of soil region $s$, which pools the responses for all counties in the soil region. Second, ISU multiplies $\bar P_{s(j),t}$ by the \emph{census relative} $\chi_j$ of county $j$, which is the ratio of the county's land value in the Census of Agriculture to the simple average of the census land values in its soil region, averaged over the five most recent censuses. The result is the unadjusted county value
\begin{equation}
    \widetilde P_{jt}=\chi_j\bar P_{s(j),t}. \label{eq:ISU_county}
\end{equation}
Third, because land markets do not stop at the borders of soil regions and land quality changes gradually across them \citep[pp.~14--15]{HarrisLordWeirich1980}, ISU replaces $\widetilde P_{jt}$ for each county that borders a different soil region by a weighted average of $\widetilde P_{jt}$ itself, with weight 2, and the unadjusted values of the bordering counties in other soil regions, with weights that sum to 1 and are proportional to the length of the common boundary. The result is the published county value $P^c_{jt}$. Figure~\ref{fig:regions_map} in the main text shows how the soil-region and district boundaries differ.

Equation \eqref{eq:ISU_county} explains why $P_t$ is almost identical to the simple average of the published county values. By construction, the census relatives $\chi_j$ average to one within each soil region, so the unadjusted county values average to the soil region value $\bar P_{st}$. If the soil region value is the simple average of the preliminary county values in the soil region, as is the case for the districts, then the simple average of the 99 unadjusted county values equals $P_t$ exactly. Only the border adjustment and rounding create a gap. Thus, the published county values are constructed to be consistent with the statewide value, not the other way around.

\paragraph{Evidence that county values move together by construction}
Call a county \emph{interior} if all counties that share a border with it lie in its own soil region. Because the border adjustment does not affect interior counties, \eqref{eq:ISU_county} implies that the published values of interior counties in the same soil region grow at the same rate except when ISU updates the census relatives. Figure~\ref{fig:county_comovement} confirms this prediction for the 18 interior counties in the six soil regions that contain at least two of them: relative to their soil region, their log values stay constant except in a few update years, most of which come about two years after a Census of Agriculture. Before 1980, the update years come every five years, which suggests that \citet{HarrisLordWeirich1980} updated the census relatives with each census when they re-estimated the values for 1950--1979. Thus, the soil regions of \citet{HarrisLordWeirich1980} still describe ISU's procedure. Table~\ref{tab:county_fe} shows that, in years without updates, soil region $\times$ year effects explain 97.4 percent of the variation in the annual changes in log county values across all counties and all of it across interior counties.

\begin{figure}[!htb]
\centering
\includegraphics[width=0.85\linewidth]{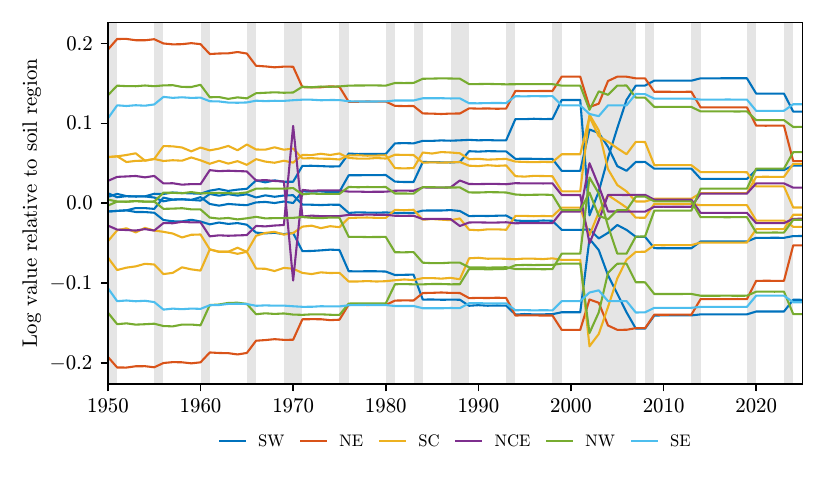}
\caption{Land values of interior counties relative to their soil region, 1950--2025.}\label{fig:county_comovement}
\caption*{\footnotesize \emph{Notes:} The figure plots, for each of the 18 interior counties in the six soil regions with at least two interior counties, the log published value minus the average log value of the interior counties in its soil region. Shaded intervals mark the update years (1951, 1956, 1961, 1966, 1970, 1971, 1976, 1981, 1984, 1988, 1989, 1994, 1999, 2002--2007, 2009, 2014, 2020, and 2024), which we define as the years in which the log value ratio of two interior counties in the same soil region changes by more than the rounding of the published values to whole dollars can explain. Rounding changes the log of a value $V$ by at most $0.5/V$. Soil region abbreviations are as in Figure~\ref{fig:regions_map}.}
\end{figure}

\begin{table}[!htb]
\centering
\caption{Share of annual changes in log county values explained by fixed effects, 1951--2025.}\label{tab:county_fe}
\begin{tabular}{lccc}
\toprule
& \multicolumn{2}{c}{All counties} & Interior counties \\
\cmidrule(lr){2-3}\cmidrule(lr){4-4}
Fixed effects & Without updates & Update years & Without updates \\
\midrule
Year & 0.929 & 0.850 & 0.924 \\
District $\times$ year & 0.968 & 0.910 & --- \\
Soil region $\times$ year & 0.974 & 0.920 & 1.000 \\
Observations & 5,148 & 2,277 & 936 \\
\bottomrule
\end{tabular}
\caption*{\footnotesize \emph{Notes:} The table reports the $R^2$ of regressions of the annual change in log county value $\Delta\log P^c_{jt}$ on the indicated fixed effects. Update years are as defined in Figure~\ref{fig:county_comovement}. Interior counties are the 18 interior counties in the six soil regions with at least two interior counties. We do not report district $\times$ year effects for interior counties, because several districts contain only one of them.}
\end{table}

This construction has three implications for our analysis. First, the published county values within a soil region move together by construction, so the 99 county series contain far fewer independent observations than their number suggests. Second, the relative values of counties within a soil region come from the Census of Agriculture with a long lag, because $\chi_j$ averages over five censuses, which span about 20 years. Third, the average of the published county values in a district differs from the published district value, because the soil regions and the border adjustment mix counties from different districts. On average over 1950--2025, the gap is less than 2 percent in seven of the nine districts but is $+9.6$ percent in South Central and $-5.0$ percent in Southeast. For these reasons, \S\ref{sec:empirical} uses the published district values.

\subsection{Cash rents}\label{subsec:data_rents}

\paragraph{Survey and statewide rent}
ISU's annual \emph{Cash Rental Rates for Iowa Survey} asks farmers, landowners, agricultural lenders, and professional farm managers for their best judgment of the typical cash rent for high-, medium-, and low-quality corn and soybean land in their counties \citep{ISUCashRent2026}. The survey does not ask about rents for individual farms. In 2026, it received 1,450 usable responses, or about 15 per county. The survey reports describe no construction step comparable to the census relatives or the border adjustment for land values, although the 2002 report notes that ISU combined the data for some similar counties when it received only a few responses \citep{ISUCashRent2002}. ISU does not report the rent of a county with fewer than five responses \citep{ISUCashRent2026}. The historical file that we use collects the overall average rent of each county and the statewide rent $r_t$ from these reports. County coverage is sparse in the early years: fewer than 50 counties are reported before 1984, and 79 to 93 counties are reported in 1987--1993. From 1994 onward, $r_t$ is almost identical to the simple average of the county rents (Appendix~\ref{subsec:data_statewide}).

\paragraph{Statewide rent before 1994}
Before 1994, $r_t$ differs from the simple average of the reported county rents. During 1978--1983, for example, their mean absolute percentage difference is 14 percent. Much of this difference reflects the five-response rule: counties with more responses tend to have more valuable land and higher rents, so the reported counties are not representative of the state. To check this explanation, suppose that $r_t$ is the simple average of the rents of all 99 counties, including those that are not reported. Let $n_t$ denote the number of reported counties and $S_t$ the sum of their rents. Then the average rent of the $99-n_t$ unreported counties must equal $(99r_t-S_t)/(99-n_t)$. If rents are roughly proportional to land values across counties, the ratio of this implied rent to $r_t$ should be close to the ratio of the average land value of the unreported counties to that of all counties, which we observe. From 1978 to 1988, the two ratios differ by at most 0.03 in every year, so $r_t$ is consistent with an average over all counties.

From 1989 to 1993, however, the implied rent of the unreported counties is 1.14 to 1.66 times $r_t$, although their land values are only 0.69 to 0.99 times the average. In 1989 and 1993, this implied rent even exceeds the highest reported county rent. In 1993, for example, the six unreported counties would need an average rent of \$179, whereas the highest reported rent is \$137. If the rent of the unreported counties relative to $r_t$ instead equaled their relative land value, as it does through 1988, the average rent of all counties would be 4 to 9 percent below $r_t$ in 1989--1993. Consistent with this gap, $r_t$ falls from \$108 in 1993 to \$102 in 1994, whereas the rents of the 93 counties reported in both years fall by only 1.5 percent on average and $P_t$ rises by 6.4 percent. The survey reports do not explain how ISU computed $r_t$ in these years.

For these reasons, our main sample for any analysis that uses rents starts in 1994, and we use 1978--2025 only as a robustness check.

\paragraph{District rents}
Letting $r^c_{jt}$ denote the published real cash rent of county $j$, we construct the rent of district $i$ as $r_{it}=\sum_{j\in i}r^c_{jt}/n_i$, the simple average of the rents of its counties, which matches how ISU constructs the district land value $P_{it}$. Since 2010, the annual survey reports also publish the average rent of each district; before 2010, they report averages for twelve different groups of counties. Over 2010--2025, our district rents differ from the published district averages by 0.3 percent on average and by at most 3.2 percent (South Central in 2024), and the correlation of their log deviations from the cross-district average is 0.999. We use our constructed series because it covers the whole sample period.

\paragraph{Imputation}
Cash rents are missing for 12 counties in 1995, 4 in 1996, and 1 each in 1998 and 1999; all 99 counties are observed in 1994, 1997, and every year from 2000 through 2025. We impute each missing county rent by $r^c_{jt}=\hat\psi^c_jr_t$, where $\log\hat\psi^c_j$ is the average of $\log(r^c_{jt}/r_t)$ over the years in which we observe county $j$. This imputation relies on a common factor in county rents, analogous to the one in district rents documented in \S\ref{subsec:common_rent}. We use the imputed rents only to construct the district rents $r_{it}$, and they enter only five district--years: South Central in 1995, 1996, 1998, and 1999 and Southeast in 1995.

We evaluate the imputation using the period 2000--2025, when we observe all counties. For each year $t$ in this period and each of the four years with missing rents, we remove the counties that are missing in that year, impute their rents using loadings $\hat\psi^c_j$ estimated without year $t$, and compare the resulting rents of the affected districts with the true district rents. The imputation error averages $-0.1$ percent, its mean absolute value is 1.2 percent, and it never exceeds 5.7 percent in absolute value. In contrast, if we simply average the remaining counties in each district, the error reaches 12.9 percent, because the missing counties tend to have lower rents than the other counties in their districts.

The imputed rents do not affect our results. If we average only the observed county rents in each district, the trend in the dispersion of $\log(P_{it}/r_{it})$ is $-0.0021$ rather than $-0.0020$. If we drop the four years with missing county rents when we estimate the district rent factor in \S\ref{subsec:common_rent}, the mean-group slope is 1.021 rather than 1.023, and the within $R^2$ with unit loadings is 0.967 rather than 0.965.

\subsection{Statewide measures}\label{subsec:data_statewide}

\paragraph{Comparison with county averages}
Let $J=99$ and define the simple county averages
\begin{equation*}
    \bar P_t\coloneqq\frac{1}{J}\sum_{j=1}^J P^c_{jt},
    \qquad
    \bar r_t\coloneqq\frac{1}{J}\sum_{j=1}^J r^c_{jt}.
\end{equation*}
Over 1950--2025, the published land value $P_t$ is almost identical to the simple average $\bar P_t$: their mean absolute percentage difference is only 0.18 percent, and the ratio $\bar P_t/P_t$ always lies between 0.9945 and 1.0054. Over 1994--2025, the published rent $r_t$ is similarly close to the simple average $\bar r_t$: their mean absolute percentage difference is 0.50 percent, and the ratio $\bar r_t/r_t$ always lies between 0.982 and 1.005. (In the four years with missing county rents, $\bar r_t$ uses the imputed rents of Appendix~\ref{subsec:data_rents}.)

\paragraph{Acreage weighting}
Our statewide and district measures weight all counties equally. To check whether this matters, let $w_j$ denote county $j$'s share of the state's corn-and-soybean acreage in 2022, approximated as described in the notes to Table~\ref{tab:data_sources}, and define the acreage-weighted measures $P_t^w\coloneqq\sum_jw_jP^c_{jt}$ and $r_t^w\coloneqq\sum_jw_jr^c_{jt}$. On average over 1994--2025, $P_t^w$ exceeds $P_t$ by 5.5 percent and $r_t^w$ exceeds $r_t$ by 2.5 percent, because counties with more corn and soybean acreage tend to have more valuable land and higher rents. However, the two cash yields are almost identical: $r_t^w/P_t^w$ is 2.31 percent in 2025 compared with 2.35 percent for $r_t/P_t$, and the correlation of their logarithms over 1994--2025 is 0.9999.

\subsection{Additional results}\label{subsec:district_results}

\paragraph{Dispersion with three measures}
Table~\ref{tab:dispersion} reports the dispersion of $\log X_{it}$ over 1994--2025 with three measures of the normalized land price. In addition to the baseline $X_{it}=P_{it}/\hat\psi_i$, we use $X_{it}=P_{it}/\hat r_{it}$, where $\hat r_{it}=\exp(\hat\alpha_i+\hat\beta_i\log r_t)$ is the fitted rent from the district regressions \eqref{eq:rent_slope}, which allow each district its own rent growth, and the observed price--rent ratio $X_{it}=P_{it}/r_{it}$. The dispersion falls by about half with all three measures. If each district's cash yield had kept its 1994--1999 difference from the cross-district average, \eqref{eq:gordon} would imply a dispersion in 2025 of 0.53 to 0.58, about five times the actual value.

\begin{table}[!htb]
\centering
\caption{Cross-district dispersion of normalized land prices, 1994--2025.}\label{tab:dispersion}
\begin{tabular}{lccc}
\toprule
Normalized land price $X_{it}$ & $P_{it}/\hat\psi_i$ & $P_{it}/\hat r_{it}$ & $P_{it}/r_{it}$ \\
\midrule
Standard deviation of $\log X_{it}$ in 1994 & 0.209 & 0.210 & 0.183 \\
Standard deviation of $\log X_{it}$ in 2025 & 0.107 & 0.106 & 0.104 \\
Trend per year & $-0.0023$ & $-0.0024$ & $-0.0020$ \\
 & (0.0005) & (0.0004) & (0.0003) \\
2025 with fixed yield gaps & & 0.583 & 0.527 \\
\bottomrule
\end{tabular}
\caption*{\footnotesize \emph{Notes:} Real land prices and rents for the nine districts, 1994--2025. The standard deviation divides by $I-1=8$, and the trend is the slope of a linear time trend, with Newey--West standard errors (three lags) in parentheses. ``Fixed yield gaps'' is the 2025 dispersion implied by \eqref{eq:gordon} if each district's cash yield kept its 1994--1999 difference from the cross-district average; it is not defined for the first column, whose denominator is not a rent.}
\end{table}

\paragraph{Heterogeneous rent slopes}
Figure~\ref{fig:rent_slopes} plots the estimates of $\beta_i$ in the district rent regressions \eqref{eq:rent_slope}. The mean-group estimate \citep{PesaranSmith1995} is 1.02 with standard error 0.04, so the average district loads one for one on the statewide rent. However, six of the nine slopes differ from one at the 5 percent level. Consistent with this heterogeneity, relative rents $r_{it}/r_t$ drift slowly: district-specific linear trends explain 43 percent of the variation in $\epsilon_{it}$, and the cumulative drift from 1994 to 2025 has a cross-district standard deviation of about 7 percent. Relative rents rose in the Northwest, Northeast, West Central, and Southwest districts and fell in the East Central and Southeast districts.

\begin{figure}[!htb]
\centering
\includegraphics[width=\linewidth]{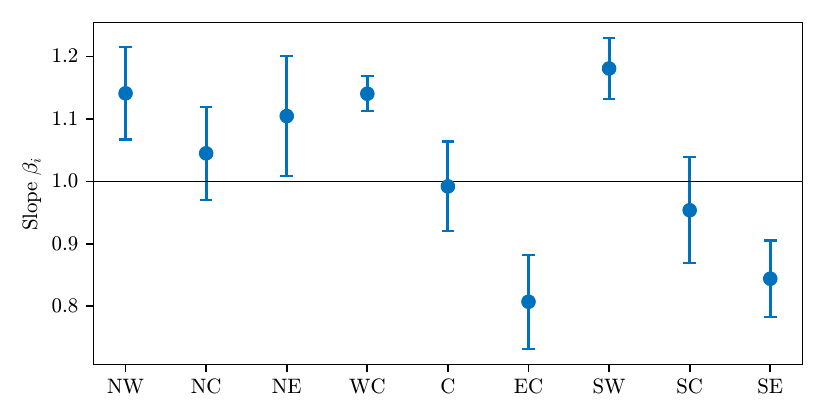}
\caption{Slopes of district rents on the statewide rent, 1994--2025.}\label{fig:rent_slopes}
\caption*{\footnotesize \emph{Notes:} The figure plots the estimate of $\beta_i$ in \eqref{eq:rent_slope} for each district with 95 percent confidence intervals based on Newey--West standard errors (three lags). District abbreviations are as in Figure~\ref{fig:regions_map}.}
\end{figure}

\paragraph{Leaving out one district}
South Central Iowa, the district with the lowest relative price in 1994, accounts for much of the decline in dispersion. Without it, the dispersion falls from 0.14 in 1994 to 0.09 in 2025, and the trend is $-0.0012$ (standard error 0.0004), about half of the baseline. Dropping any other single district leaves the trend between $-0.0021$ and $-0.0033$. In all cases, the trend is at least 2.8 times its standard error.

\subsection{Additional details for the Gordon model}\label{subsec:gordon_details}

\paragraph{Aggregate uncertainty}
The Gordon formula \eqref{eq:gordon} extends to aggregate uncertainty. Suppose that the stochastic discount factor follows $\diff M_t/M_t=-r^f\diff t-\lambda\diff W_{1t}$ and the rent of district $i$ follows $\diff r_{it}/r_{it}=g_i\diff t+\sigma\diff W_{2t}$, where $W_{1t}$ and $W_{2t}$ are Brownian motions with correlation $\varphi$. Then
\begin{equation*}
	V_{it}=\E_t\int_0^\infty(M_{t+s}/M_t)r_{i,t+s}\diff s=\frac{r_{it}}{r^f+\lambda\sigma\varphi-g_i},
\end{equation*}
which is \eqref{eq:gordon} with required return $\rho=r^f+\lambda\sigma\varphi$. Because all districts share the stochastic discount factor and, by \eqref{eq:rent_factor}, the rent risk, the risk premium $\lambda\sigma\varphi$ is common across districts.

With aggregate uncertainty, no arbitrage implies only that each district bubble satisfies $B_{it}=\E_t[M_{t+1}B_{i,t+1}]$ for a common stochastic discount factor $M_{t+1}$. This condition neither makes district bubbles proportional nor requires the realized bubble to grow at the realized return on land every year. Our models therefore impose two restrictions motivated by the deterministic theory: a common factor drives all district bubbles, and in model~\modelBR this factor earns the realized return on land.

\paragraph{Bootstrap robustness}
Estimating the autoregressive coefficients from the residuals of the more general model instead of the restricted model gives 5 percent critical values of 81, 29, and 69 for the three tests in Table~\ref{tab:gordon_tests}, compared with 85, 47, and 68, and does not change any conclusion.

\paragraph{Property taxes and the start date}
If property tax takes a share $\theta$ of rent, the after-tax return on land is lower than $R_t$ by $\theta r_{t+1}/P_t$. With $\theta=0.1$ or $0.2$, the quasi-likelihood still selects $t_0=2002$, the bubble accounts for 49 or 50 percent of the average land price in 2025, and the quasi-likelihood ratio of model~\modelBR against model~\modelBU is 55 or 61, below the 5 percent critical value in Table~\ref{tab:gordon_tests}. With $\theta=0.3$, the quasi-likelihood selects $t_0=1994$, the bubble accounts for 67 percent of the average land price in 2025, and the quasi-likelihood ratio is 58.

The quasi-likelihood is fairly flat near $t_0=2002$: relative to that date, the quasi-likelihood ratio is 3.4 for 2001 and 5.8 for 1999, but it exceeds 10 for any other start date. Because model~\modelBR makes the bubble appear at once, its estimated bubble jumps to 24 percent of the average land price in 2002, whereas the unrestricted bubble in model~\modelBU rises from 5 percent in 1995 to 25 percent in 2002. The longer samples select 1995 and 2002 (Appendix~\ref{subsec:gordon_long}).

\subsection{Gordon model with a measured real interest rate}\label{subsec:gordon_observed}

This appendix shows that the Gordon model fails when the capitalization rate follows a measured real interest rate. Let $\rho_t$ be the 10-year Treasury yield minus the 10-year expected inflation of the Federal Reserve Bank of Cleveland, in November of each year, which is available from 1982. Over 1982--2025, we estimate \eqref{eq:gordon_regression} with $\rho_t$ equal to this rate and with two changes. First, the district constant $\delta_i$ now also absorbs the risk premium and expected rent growth. Second, because $\rho_t$ is observed, we allow the bubble to differ across districts by replacing $B_t/r_t$ with $\gamma_iB_t/r_t$, where the loadings $\gamma_i\ge0$ average one. The bubble $B_t\ge0$ is free in every year, and the error $u_{it}$ follows an AR(1) process with a common coefficient. We estimate the parameters by exact Gaussian maximum likelihood.

Without a bubble ($B_t=0$), the estimated autoregressive coefficient reaches its upper bound of 0.99, so the fit errors behave almost like a random walk: the measured real interest rate cannot explain the movements of land prices relative to rents. With a free bubble and no bound on $\delta_i$, the likelihood sends the fundamental value toward zero and attributes about 90 percent or more of the average land price to the bubble, which only restates this failure. When we bound $\delta_i$ by 0.10, a parametric bootstrap rejects the specification without a bubble against the one with a free bubble at any conventional level (likelihood ratio 1,127, 5 percent critical value 95).
These results illustrate the excess volatility of \citet{Shiller1981} for farmland and motivate the free capitalization rate in \S\ref{subsec:gordon_estimates}.

\subsection{Gordon model over longer samples}\label{subsec:gordon_long}

This appendix repeats the estimation of \eqref{eq:gordon_regression} over 1987--2025 and 1978--2025. Both samples use the statewide rent $r_t$ before 1994, which is less reliable (Appendix~\ref{subsec:data_rents}), together with the loadings $\hat\psi_i$ estimated from the 1994--2025 rents. The 1987--2025 sample starts after the farm crisis. Table~\ref{tab:gordon_long} reports the results together with those for the main sample.

\begin{table}[!htb]
\centering
\caption{Gordon model with and without a land bubble over different samples.}\label{tab:gordon_long}
\begin{tabular}{lccc}
\toprule
& 1994--2025 & 1987--2025 & 1978--2025 \\
\midrule
Start year $t_0$ & 2002 & 1995 & 2002 \\
\midrule
\multicolumn{4}{l}{\emph{Residual s.d.}} \\
Model~\modelF & 0.074 & 0.093 & 0.098 \\
Model~\modelBU & 0.036 & 0.035 & 0.043 \\
Model~\modelBR & 0.039 & 0.040 & 0.056 \\
\midrule
\multicolumn{4}{l}{\emph{Quasi-likelihood ratio (5\% critical value)}} \\
\modelF against \modelBU & 420 (85) & 682 (110) & 715 (126) \\
\modelF against \modelBR & 369 (47) & 600 (66) & 484 (60) \\
\modelBR against \modelBU & 51 (68) & 82 (83) & 231 (85) \\
$p$-value, \modelBR against \modelBU & 0.17 & 0.053 & 0.001 \\
\midrule
\multicolumn{4}{l}{\emph{Bubble share in 2025}} \\
Model~\modelBU & 0.59 & 0.63 & 0.64 \\
Model~\modelBR & 0.50 & 0.62 & 0.52 \\
\bottomrule
\end{tabular}
\caption*{\footnotesize \emph{Notes:} Estimates of \eqref{eq:gordon_regression} as in Table~\ref{tab:gordon_tests}, with $T=32$, 39, and 48. The $p$-values of all tests of model~\modelF are 0.001. The bubble share is $(B_t/r_t)/\bar Q_t$, where $\bar Q_t$ is the cross-district average of $Q_{it}$.}
\end{table}

In both longer samples, the bootstrap rejects model~\modelF against both bubble models at any conventional level, and the start year and the size of the bubble are similar to those in the main sample. Over 1978--2025, model~\modelF requires the capitalization rate of the average district to rise from 5.1 percent in 1978 to 10.8 percent in 1986 and then fall to 2.5 percent in 2025. In contrast, in model~\modelBR, the capitalization rate is 5.3 percent in 2025, close to its value of 5.4 percent in 1978 (Figure~\ref{fig:gordon_long}(c)).

Over 1978--2025, the bootstrap rejects model~\modelBR against model~\modelBU. Figure~\ref{fig:gordon_long}(b) shows the reason: model~\modelBU finds an earlier episode. Its unrestricted bubble accounts for 27 percent of the average land price in 1978 and 31 percent in 1981, and it nearly vanishes by 1986, as the average price--rent ratio falls from 20.8 in 1981 to 9.4 in 1986 during the farm crisis. Model~\modelBR cannot capture this episode. A bubble that earns the land return never shrinks relative to the land price: by \eqref{eq:bubble_growth}, its ratio to the land price grows at about the cash yield. A bubble that collapses lies outside our deterministic theory and would require, for example, a stochastic bubble that bursts with some probability \citep{BlanchardWatson1982,Weil1987}.
Over 1987--2025, which excludes this episode, the evidence against model~\modelBR is borderline: the bootstrap rejects it at the 10 percent level but not at the 5 percent level ($p=0.053$).

\begin{figure}[!htb]
\centering
\includegraphics[width=\linewidth]{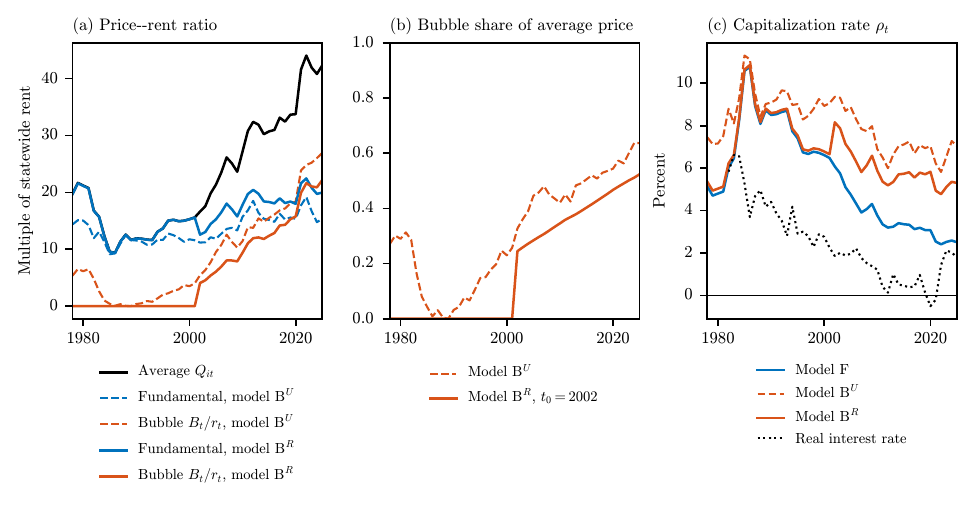}
\caption{Fundamental value and bubble of Iowa farmland, 1978--2025.}\label{fig:gordon_long}
\caption*{\footnotesize \emph{Notes:} Estimates of \eqref{eq:gordon_regression} over 1978--2025. The panels are as in Figure~\ref{fig:gordon_main}. The expected-inflation series in panel (c) starts in 1982.}
\end{figure}

Finally, the results over 1978--2025 change little when we treat the statewide rent before 1994 differently. We either replace $r_t$ before 1994 with a rent index that averages the log rents of the reported counties after removing their loadings $\log\hat\psi^c_j$ (Appendix~\ref{subsec:data_rents}), or let the level of $r_t$ before 1994 differ by an estimated factor. In both cases, $t_0$ is 2001 or 2002, the bubble share in 2025 in model~\modelBR is between 51 and 54 percent, and the quasi-likelihood ratio of model~\modelBR against model~\modelBU stays between 190 and 231.

\printbibliography

\end{refsection}

\end{document}